\documentclass[journal]{IEEEtran}

\IEEEoverridecommandlockouts

\usepackage{caption}
\usepackage{subcaption}
\usepackage{graphicx}
\usepackage{cite}
\usepackage{url}
\usepackage[cmex10]{amsmath}
\usepackage{amssymb}
\usepackage{amsthm}
\usepackage{algorithm}
\usepackage{algorithmicx}
\usepackage{algpseudocode}
\usepackage{cases}
\usepackage{color}
\usepackage{soul}
\usepackage{xcolor}
\usepackage{framed}
\usepackage{cite}
\usepackage{epstopdf}
\usepackage{calc}
\usepackage{array}
\usepackage{amsmath}
\usepackage{comment}
\usepackage{bbm}
\colorlet{shadecolor}{yellow}
\DeclareGraphicsExtensions{.pdf,.eps,.png,.jpg,.mps}

\newtheorem{theorem}{Theorem}
\newtheorem{corollary}{Corollary}

\begin{document}
\title{Power Control and Frequency Band Selection Policies for Underlay MIMO Cognitive Radio}

%
\author{Shailesh Chaudhari, \textit{Student Member, IEEE},  Danijela Cabric, \textit{Senior Member, IEEE}%
\thanks{Shailesh Chaudhari and Danijela Cabric are with the Department of Electrical and Computer Engineering, University of California, Los Angeles, 56-125B Engineering IV Building, Los Angeles, CA 90095-1594, USA (email: schaudhari@ucla.edu, danijela@ee.ucla.edu).}
\thanks{This work has been supported by the National Science Foundation under CNS grant 1149981.}
}
%
%
%


\maketitle

\begin{abstract}
We study power control and frequency band selection policies for multi-band underlay MIMO cognitive radio with the objective of maximizing the rate of a secondary user (SU) link while limiting the {interference leakage} towards primary users (PUs) below a threshold. The goal of the SU in each policy is to select one frequency band in each time slot and determine the transmit power. To limit the interference towards PU in time-varying channels, we propose fixed and dynamic transmit power control schemes which depend on PU traffic and the temporal correlation of channels between the SU and the PU. We study the performance of frequency band selection policies that use fixed or dynamic power control. We show that dynamic frequency band selection policies, e.g., policies based on multi-armed bandit framework, wherein SU selects a different frequency band in each slot, result in higher interference towards PU as compared to the fixed band policy wherein SU stays on one band. We also provide an expression for the gap between the rate achieved by SU under a clairvoyant policy and the fixed band policy. It is observed that this gap reduces with increased temporal correlation and with increased number of SU antennas.
\end{abstract}

\IEEEpeerreviewmaketitle



\section{Introduction}
\label{sec:Introduction}
Due to ever increasing usage of mobile devices and data hungry applications, it has become necessary to improve the spectral efficiency of wireless networks. The spectral efficiency can be improved by allowing co-existence of unlicensed secondary users (SUs) with licensed primary users (PUs) in the same frequency band. Cognitive radio (CR) networks allow such co-existence under two paradigms: {interweave} and underlay \cite{biglieri2012a, tanab2017}. In an {interweave} CR network, the SU can transmit only in \textit{empty} time-slots when PUs are inactive in order to avoid interfering with the them. The achievable rate of the SU is further improved if there are multiple frequency bands available for transmission. {In a multi-band {interweave} CR network}, the SU can maximize its achievable rate by predicting which frequency band will have an empty time slot and then tunning to that band for transmission. Thus, the SU can dynamically hop to a different frequency band in each time-slot in search of an empty time slot to maximize its rate. This prediction-based band hopping is achieved by the multi-armed bandit (MAB) framework \cite{zhao2008a, tekin2011, liu2013a, dai2014, ouyang2014, oksanen2015, wang2016a, maghsudi2016, raj2018}. In the MAB framework, the SU learns the on-off activity of PUs in different bands in order to predict which band (arm) will be empty in the next time slot. However, the achievable rate of the SU in the {interweave} network is limited by the PU activity since the probability finding an empty slot is low when the PU activity is high. Further, a costly RF front-end is required at the SU to hop to a different band in each slot.

The achievable rate of SU can be improved if it is allowed to transmit even when the PU is active. The underlay CR paradigm allows the SU to transmit concurrently with PUs as long as the the interference towards primary receiver is below a specified limit \cite{biglieri2012a}. 
The SU can transmit concurrently with the PU, if the SU is equipped with multiple antennas and employs beamforming techniques to steer its signal in the null space of channels to primary receiver in order to contain the interference \cite{tsinos2013, noam2013}. The null space to primary receiver is estimated using the received auto-covariance matrix at the SU during a previous slot when receiver was the transmitter \cite{tsinos2013, gao2010, yi2009, yi2010}. Since the channel between SU and PU evolves due to fading during these time slots, the SU cannot not perfectly eliminate the interference towards the primary receiver using only null steering. Therefore, transmit power control is required along with null steering to limit the interference. The transmit power from the SU depends on the time between transmit and receive modes of PUs, i.e., the link reversal time of the PU link. In other words, the power transmitted from SU depends on the traffic pattern of the PU transmitter-receiver link as well as the temporal correlation that determines the rate of channel fading. Such transmit power control has not been considered in underlay MIMO CR literature and is addressed in this work. In an underlay CR network, the rate of the SU link depends on the transmit power as well as beamforming gain achieved after null steering. Therefore, {for a multi-band underlay CR network}, the band selection policy needs to take into account transmit power, beamforming gain and PU traffic statistics in each frequency band.
		
\subsection{Related Work}
Frequency band selection using MAB based prediction has been considered for interweave cognitive radio in \cite{zhao2008a, tekin2011, liu2013a, dai2014, ouyang2014, oksanen2015, wang2016a}. In these works, the problem is cast as a restless MAB where each frequency band is modeled as an independent arm of the bandit problem. The term \textit{restless} implies that the physical channels in each band keep evolving even when that band is not selected by the SU, which holds for wireless channels. The goal of the band selection policies using restless MAB is to maximize the expected rate at the SU. Since the optimal solution to a general restless MAB problem is intractable \cite{zhao2008a, ouyang2014}, most of the works consider special cases. The special cases include policies based on a binary channel model as well as myopic policies where the goal is to maximize immediate rate in the next time slot. In the works \cite{zhao2008a, tekin2011, dai2014, oksanen2015}, a binary channel model was considered, where the SU receives reward (rate) 0 if the selected band is occupied by the PU, otherwise it receives rate 1. This model is suitable in the interweave CR network where SU transmits only when PU is inactive. For the binary channel model, the optimality of myopic band selection policy was shown in \cite{zhao2008a} for two frequency bands under the condition that the channel state evolves independently from one slot to the next. An online learning based band selection was proposed in \cite{dai2014} that implements the myopic policy in \cite{zhao2008a} without prior knowledge of PU activity. The work in \cite{tekin2011} considered a more general case where the state of the binary channel is modeled as a Markov chain (Gilbert-Elliot model). The authors proposed regenerative cycle algorithm (RCA) that outperforms the selection scheme in \cite{zhao2008a}. A recency based band selection policy was introduced in \cite{oksanen2015}, where the SU selects a suboptimal band less frequently and thus provides better performance as compared to earlier policies in a binary channel model with independent or Markovian evolution.

Binary channel models are not suitable for band selection in underlay CR network where the SU can receive a non-zero rate even when PU is active in the selected band. The rate received in this case depends on the beamforming gain between secondary transmitter and receiver as well as the transmit power. In order to model the beamforming gain, a multi-state channel model is required. Multi-state channels are considered for restless MAB problems in \cite{ouyang2014, wang2016a, liu2013a}. The optimality of myopic policy is established in \cite{ouyang2014} for a multi-state channel under the condition that the rate received by the SU in different channel states is sufficiently separated. This condition, however, may not hold in a real world channel with continuous state space. The work in \cite{wang2016a} established the optimality of the myopic policy when $F-1$ out of $F$ channels are selected by SU in each time slot. A policy, called deterministic sequencing of exploration and exploitation (DSEE) was constructed in \cite{liu2013a}. Under this policy, the SU stays on one band for multiple consecutive slots, called epochs, and the epoch length grows geometrically. It has been shown that the DSEE outperforms RCA for multi-state channels.

In an underlay CR network, if the transmit power is known, then the frequency band selection policy can be constructed by aforementioned restless MAB approaches such as DSEE. However, the existing works do not consider transmit power control for such a band selection problem in underlay CR networks. 
\subsection{Summary of Contributions and Outline}
In this paper, we first propose fixed and dynamic power control schemes for a SU with multiple antennas. In the fixed power scheme, the SU transmits fixed power when the PU is active in that time slot, while in the dynamic power control scheme, the transmit power from the SU changes in each time slot. We show that the transmit power in a given frequency band depends on the traffic statistics of the PU transmitter-receiver links and the temporal correlation of the channels. 

Next, we analyze the following categories of band selection policies that use the above power control schemes: fixed band fixed power (FBFP), fixed band dynamic power (FBDP), dynamic band fixed power (DBFP) and clairvoyant policy. In the FBFP and FBDP policies, the SU stays on one frequency band and uses fixed or dynamic power control. The band selection policies based on restless MAB, such as DSEE, fall under DBFP category along with round robin and random band selection policies. The SU may hop to different frequency band under the DBFP policies. We also analyze the performance of a genie-aided clairvoyant policy that selects the frequency band providing the maximum gain in each slot. We compare the performance of these policies in terms of rate received at SU and interference towards PU.

The main contributions of this paper are summarized below.
\begin{enumerate}
	\item Expressions for transmit power are derived for fixed and dynamic power control schemes as functions of link reversal time of the PU transmitter-receiver link and temporal correlation of channels. It is observed that the transmit power and thus the rate of the SU increase as the PU link reversal time decreases. 
	\item We show that the dynamic power control policy provides higher rate to SU than the fixed power control, i.e., FBDP provides higher rate than FBFP. Both polices keep the interference leakage towards PU below the specified limit. It is also shown that the DBFP polices, such as round robin, random, and DSEE, cause higher interference to PUs as compared to the fixed band policies.
	\item The expression is derived for the gap between the rate achieved by an optimal genie-aided clairvoyant policy and the FBFP policy. It is shown that this gap reduces under slow-varying channels and as the number of SU antennas is increased. This implies that the SU does not loose significant amount of rate by staying on one frequency band.
\end{enumerate}

\textit{Outline:} This paper is organized as follows. The system model and problem statement are described in Section \ref{sec:Model_problem}. Power control and band selection policies are discussed in Section \ref{sec:policies}. Analytical comparison of the policies is presented in Section \ref{sec:analysis} while numerical results are shown in Section \ref{sec:results}. Finally, concluding remarks are provided in Section \ref{sec:Conclusion} and future extension is discussed in Section \ref{sec:extension}.

\textit{Notations:} We denote vectors by bold, lower-case letters, e.g., $\mathbf{h}$. Matrices are denoted by bold, upper case letters, e.g., $\mathbf{G}$. Scalars are denoted by non-bold letters e.g. $L$. Transpose, conjugate, and Hermitian of vectors and matrices are denoted by $(.)^T$, $(.)^*$, and $(.)^H$, respectively. The norm of a vector $\mathbf{h}$ is denoted by $||\mathbf{h}||$. $\Gamma(x)$ is the Gamma function, while $\gamma(M,x)$ is the incomplete Gamma function defined as $\int_{0}^{x}t^{M-1}e^{-t} dt$. $\mathbb{E}[.]$ denotes the expectation operator, while $\mathbb{E}_x[.]$ is the expectation with respect to random variable $x$.


\begin{figure}
	\centering
	\includegraphics[width=\columnwidth]{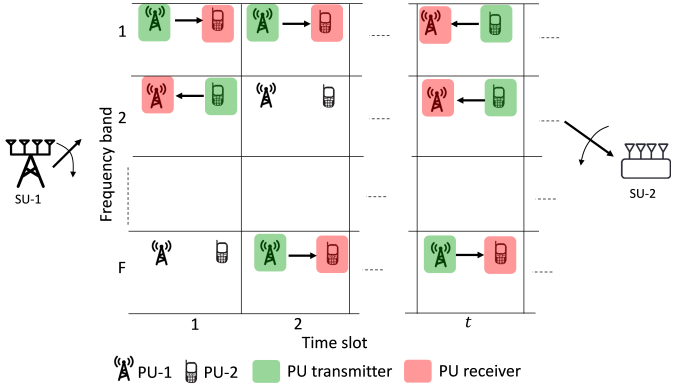}
	\vspace{-2mm}
	\caption{\small System model: SUs select any one out of $F$ available bands at time slot-$t$. Each band is occupied by a PU transmitter-receiver link.}
	\label{fig:system_model}
	\vspace{-6mm}
\end{figure}

\section{System Model and Problem Formulation}
\label{sec:Model_problem}
\subsection{System Model}
Consider an underlay CR network in which SU transmit-receive pair, SU-1 and SU-2, selects one out of $F$ available frequency bands as shown in Fig. \ref{fig:system_model}. Let $M_s$ be the number of antennas at SU-1 and SU-2. Each band is occupied by a pair of PUs as shown in Fig. \ref{fig:network_model}. Let $M_p (<M_s)$ be the number of antennas at PUs. The PU-1 and PU-2 in each pair change role from transmitter to receiver according to a Markov chain.

The PU link in frequency band $f$ is in one of the three states at time slot $t$: 1) state-0: both PUs are silent, 2) state-1: PU-1 is the transmitter and PU-2 is the receiver,  and 3) state-2: PU-1 is the receiver and PU-2 is the transmitter. The state of the PU link is denoted by $s_{f,t} \in \{0,1,2\}$. The transition between the states is determined by transition probability matrix $\mathbf{T}_{f}$ as shown below:

{\small \begin{align}
\mathbf{T}_{f} = 
\begin{bmatrix}
p_{00,f} & p_{01,f} &  p_{02,f} \\
p_{10,f}& p_{11,f}  & p_{12,f} \\
p_{20,f} & p_{21,f} &  p_{22,f} \\
\end{bmatrix},
\label{eq:TPF}
\end{align}}

\vspace{-1mm}
\noindent where $p_{kl,f} = \Pr(s_{f,t+1}=l|s_{f,t}=k), k,l \in \{0,1,2\}$ is the probability that PU link goes from state-$k$ in slot $t$ to state-$l$ in slot $t+1$. The steady state probability of PU link being in state $k$ in any slot $t$ is denoted by $\pi_{k,f} = \Pr(s_{f,t}=k), k \in \{0,1,2\}$ such that $\sum_k \pi_{k,f} = 1, \forall f$. The matrix $\mathbf{T}_f$ depends on the traffic configuration of the PU link. In order evaluate the policies, we will consider TDD LTE traffic models specified 3GPP 36.211 \cite{3gpp2017} to construct $\mathbf{T}_f$. Without the loss of generality, we consider the SU-1 is the transmitter and SU-2 is receiver in the secondary network.

Consider that the SU selects band $f$ in slot $t$ and the PU link is in state $s_{f,t}=1$. Then, the channel between SUs and PUs are shown in Fig. \ref{fig:network_model}: $\mathbf{H}_{f,t} \in \mathbb{C}^{M_s\times M_s}$ is the channel between SU-1 and SU-2, while $\mathbf{G}_{ij,f,t} \in \mathbb{C}^{M_s\times M_p}, i,j \in \{1,2\}$ denote channel between PU-$i$ and SU-$j$ in time slot $t$. We assume that the channels remain unchanged for the duration of time slot and evolve from slot $t$ to slot $t+1$ according to the Gauss-Markov model as follows \cite{sadeghi2008, so2015a}: 
\begin{align}
\mathbf{H}_{f,t+1} =\alpha_f \mathbf{H}_{f,t} +\sqrt{1-\alpha_f^2} \Delta\mathbf{H}_{f,t},
\\\mathbf{G}_{ij, f,t+1} =\alpha_f \mathbf{G}_{ij,f,t} +\sqrt{1-\alpha_f^2} \Delta\mathbf{G}_{ij,f,t}, 
\label{eq:gauss_markov}
\end{align}
where $\alpha_f = J_0(2\pi f_d T_{slot})$ is the temporal correlation coefficient, $J_0(.)$ is the 0th order Bessel function, $f_d$ is the Doppler rate, and $T_{slot}$ is the duration of slot.  The matrices $\Delta\mathbf{H}_{f,t}, \Delta\mathbf{G}_{ij,f,t} \sim \mathcal{CN}(0, \mathbf{I})$ are i.i.d. channel update matrices in slot $t$. We assume that the channels are reciprocal. Initial distributions of the MIMO channels are $\mathbf{H}_{f,0} \sim \mathcal{CN}(0, \mathbf{I}) $ and $\mathbf{G}_{f,ij,0} \sim \mathcal{CN}(0, \mathbf{I})$\footnote{We consider a normalized channel model with identity covariance matrix for each flat fading MIMO channel as also used in \cite{dey2014, ho2017}. Distance based path-loss is not modeled since it does not affect the null space of channels $\mathbf{G}_{ij,f,t}$.}.

\begin{figure}
	\centering
	\includegraphics[width=0.7\columnwidth]{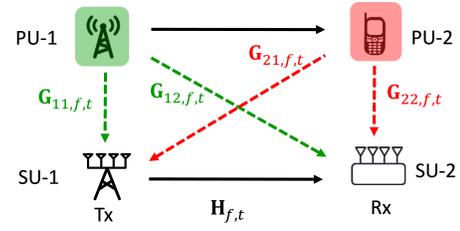}
	\vspace{-3mm}
	\caption{\small MIMO channels in frequency band-$f$ in slot $t$ with $s_{f,t}=1$. The channels in green are between PU transmitter and SUs, while the ones in red are between PU receiver and SUs.}
	\label{fig:network_model}
	\vspace{-3mm}
\end{figure}

\subsubsection{Null space computation}
\label{sec:null_space_computation}
The SU pair employs transceiver beamforming to transmit its signal in the null space  of channels to PUs. This ensures that the interference from SU transmitter to PU receiver and PU transmitter to SU receiver is minimized. The null space of the channels to PUs is obtained during the sensing duration $T_{sense}$ of each time slot. As shown in Fig. \ref{fig:slot_structure}, one time-slot consists of sensing duration $T_{sense}$ to obtain null spaces and SU data transmission $T_{data}$\footnote{The time slots in primary and secondary systems are assumed to be synchronized \cite{kaushik2016}.}. In the sensing duration, SU-1 and SU-2 receive the signal from the PU transmitter and compute the null space of the channel.

\begin{figure}
	\centering
	\includegraphics[width=0.7\columnwidth]{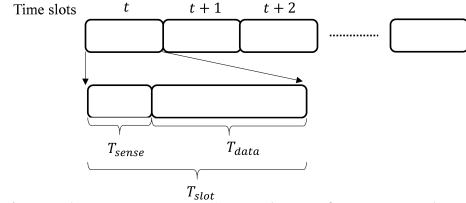}
	\vspace{-3mm}
	\caption{\small Time slot structure. SU selects frequency band $f$ at the beginning of slot and stays on that band for the duration $T_{slot}$.}
	\label{fig:slot_structure}
	\vspace{-3mm}
\end{figure}

Consider that SU selects band $f$ in slot $t$ when PU link state is $s_{f,t}=2$, i.e., PU-2 is the transmitter and PU-1 is the receiver. Let us assume that PU-1 was the transmitter $\tau_{f,t}$ slots ago, i.e., $s_{f,t-{\tau}_{f,t}}=1$. For simplicity of notations, we drop subscripts from $\tau_{f,t}$. As shown in Fig. \ref{fig:snt}, in slot $t-\tau$, SUs obtain the null space of the channels $\mathbf{G}_{11,f,t-\tau}$ and $\mathbf{G}_{12,f,t-\tau}$ using the received autocovariance matrices. Let $\mathbf{y}_{1,f,t-\tau}(n)$ and $\mathbf{y}_{2,f,t-\tau}(n)$ be the received signal vectors at SU-1 and SU-2 during sensing duration of slot $t-\tau$ as expressed below:
\begin{align}
\mathbf{y}_{1,f,t-\tau}(n) = \mathbf{G}_{11,f,t-\tau} \mathbf{x}_1(n) + \mathbf{w}(n), n = {0,2,\cdots, N-1},
\\ \mathbf{y}_{2,f,t-\tau}(n) = \mathbf{G}_{12,f,t-\tau} \mathbf{x}_1(n) + \mathbf{w}(n), n = {0,2,\cdots,N-1},
\end{align}
where $\mathbf{x}_1(n)\in \mathbb{C}^{M_p\times 1}$ is the transmitted signal vector from PU-1, $\mathbf{w}(n)\sim \mathcal{CN}(0, \mathbf{I})$ is the noise vector, and {$N=\frac{T_{sense}}{T_s}$ is the number of samples collected and $T_s$ is the sampling duration}. Let $\mathbf{A}_{1,f,t-\tau} \in \mathbb{C}^{M_s \times (M_s - M_p)}$ be the null space matrix of channel $\mathbf{G}_{11, f, t-\tau}$ and $\mathbf{B}_{1,f,t-\tau} \in \mathbb{C}^{M_s \times (M_s - M_p)}$ be the null space matrix of channel $\mathbf{G}_{12, f, t-\tau}$. Matrix $\mathbf{A}_{1,f,t-\tau}$ contains columns in the null space of the received covariance matrix $\mathbf{\hat{Q}}_{1,f,t-\tau} = \frac{1}{N}\sum\limits_{n=0}^{N-1} \mathbf{y}_{1,f,t-\tau}(n)\mathbf{y}^H_{1,f,t-\tau}(n)$ and are obtained by eigenvalue decomposition (EVD) of $\mathbf{\hat{Q}}_{1,f,t-\tau}$ at SU-1. Similarly,  matrix $\mathbf{B}_{1,f,t-\tau}$ contains columns in the null space of the covariance matrix $\mathbf{\hat{Q}}_{2,f,t-\tau} =\frac{1}{N} \sum\limits_{n=0}^{N-1} \mathbf{y}_{2,f,t-\tau}(n) \mathbf{y}^H_{2,f,t-\tau}(n)$ and are obtained by eigenvalue decomposition (EVD) of $\mathbf{\hat{Q}}_{2,f,t-\tau}$ at SU-2. Similarly, in slot $t$, SU-1 obtains null space $\mathbf{A}_{2,f,t}$ of channel matrix $\mathbf{G}_{21,f,t}$, while SU-2 obtains null space $\mathbf{B}_{2,f,t}$ of channel matrix $\mathbf{G}_{22,f,t}$. { Since the computation of null space requires EVD of a $M_s \times M_s$ matrix, it has computation complexity of $\mathcal{O}({M_s^3})$ \cite{arakawa2003}}. {SUs can estimate the state of PU link, $s_{f,t}$, based on the received signal from PUs in the sensing duration as described in Appendix \ref{app:st_est}.}

\begin{figure}[t!]
	\centering	
	\begin{subfigure}[b]{\columnwidth}
		\centering	
		\includegraphics[width=0.6\columnwidth]{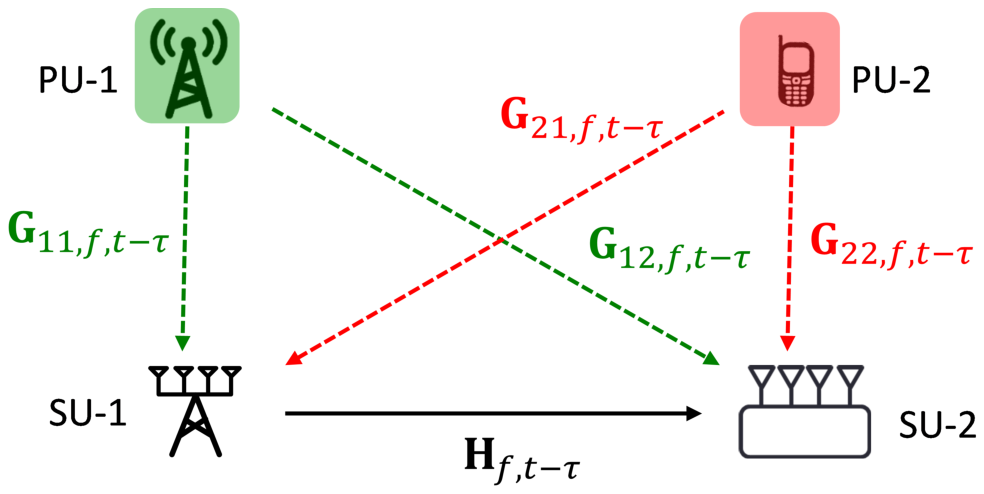}
		\caption{\small Slot $t-\tau$: ($s_{f,t-\tau}=1$) PU-1 is transmitter.}
		\label{fig:snt_t_tau}
	\end{subfigure}			
	\begin{subfigure}[b]{\columnwidth}
		\centering	
		\includegraphics[width=0.6\columnwidth]{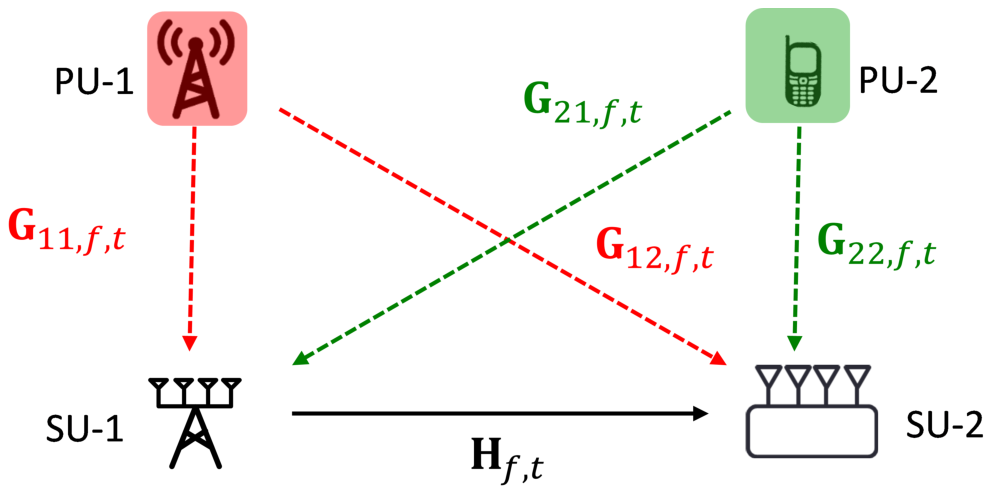}
		\caption{\small Slot $t$:  ($s_{f,t}=2$) PU-2 is transmitter.}
		\label{fig:snt_t}
	\end{subfigure}		
	\caption{\small Null space computation during sensing duration. SUs compute null space of channels shown in green by sensing the signal received from PU transmitter.}
	\label{fig:snt}	
	\vspace{-4mm}
\end{figure}

\subsubsection{Transceiver beamforming at SUs}
\label{sec:transceiver_beamforming}
Let us consider SU signal transmission from SU-1 to SU-2 in slot $t$. As shown in Fig. \ref{fig:snt_t}, in this slot, PU-2 is the transmitter and PU-1 is the receiver. In order to mitigate the interference towards PU receiver (PU-1), SU-1 needs to transmit its signal in the null space $\mathbf{A}_{1,f,t}$ of $\mathbf{G}_{11,f,t}$. However, this null space is not available at SU-1 in slot $t$, since PU-1 is not transmitting. Therefore, SU-1 utilizes the null space  $\mathbf{A}_{1,f,t-\tau}$ for precoding which was obtained in slot $t-\tau$ when PU-1 was the transmitter. On the other hand, on the receiver side, SU-2 utilizes the null space $\mathbf{B}_{2,f,t}$ for receiver combining to mitigate the interference from PU-2. In order to maximize the beamforming gain of SU link, the SU utilizes the maximum eigenmode of the equivalent channel $\mathbf{H}_{eq,f,t} = \mathbf{B}_{2,f,t}^H\mathbf{H}_{f,t}\mathbf{A}_{1,f,t-\tau}$. If the transmitted power  in slot $t$ is denoted by $P_t$, the achievable rate of the SU link is given by

\begin{align}
R_{f,t} = \frac{T_{data}}{T_{slot}}\log_2\left(1 + \frac{ P_t \Gamma_{f,t}}{\sigma^2_w}\right) = \frac{T_{data}}{T_{slot}}\log_2\left(1 + { P_t \Gamma_{f,t}}\right),
\label{eq:rate}
\end{align}
where $\Gamma_{f,t} $ is the maximum eigenvalue of $\mathbf{H}_{eq,f,t}$ and $\sigma^2_w=1$ is the noise power, and $\frac{T_{data}}{T_{slot}}$ is the fraction of time slot used for signal transmission \cite{raj2018}. Note that if both the PUs are silent in slot $t$, i.e., $s_{f,t}=0$, we get $\mathbf{H}_{eq,f,t}=\mathbf{H}_{f,t}$. The rank of $\mathbf{H}_{eq,f,t}$ is $M_s$ when $s_{f,t}=0$, while the rank is $M_s-M_p$ when  $s_{f,t}=\{1,2\}$. 

\subsubsection{Interference leakage towards PU receiver}
\label{sec:interference_to_pu}
As mentioned above, the SU transmitter cannot access the null space of channel to PU receiver in slot $t$. It uses the null space obtained $\tau$ slots ago when the PU receiver was the transmitter. This results in non-zero interference towards the PU receiver. {The expected interference leakage towards PU receiver in slot $t$ under the Gauss-Markov evolution can be written as follows:
\begin{align}
\mathbb{E}[I_{f,t}]= \mathbb{E}\left[ P_t ||\mathbf{G}^H_{11,f,t}\mathbf{v}_t||^2 \right] = P_t M_p (1-\alpha_f^{2\tau}),
\label{eq:avg_int}
\end{align}
where  $\mathbf{v}_t = \mathbf{A}_{1,f,t-\tau} \mathbf{u}_t$ is the transmit beamforming vector, $\mathbf{u}_t$ is the principle right singular vector of $\textbf{H}_{eq,f,t}$ and $M_p$ is the rank of the channel $\mathbf{G}_{11,f,t}$}. The proof of the second equality is shown in Appendix \ref{app:avg_int}. From the above expression, we observe that the interference to PU in frequency bands with smaller correlation $\alpha_f$ will be higher. Further, higher value of $\tau$ implies larger interference to PU, i.e., older the null space to PU receiver, higher will be the interference.  

Note that the expectation in (\ref{eq:avg_int}) is with respect to the random variations $\Delta \mathbf{G}_{11}$ in the Gauss-Markov model. In this expression, $\tau$ is assumed to be constant. The variable $\tau$ indicates \text{how old} the null space is at any given slot $t$. The expected value of the interference with respect to the random variable $\tau$ can be written as below:

\begin{align}
\mathbb{E}_{\tau}\left[ \mathbb{E}[I_{f,t} | \tau] \right]=  \mathbb{E}_{\tau} \left[P_t M_p (1-\alpha_f^{2\tau})\right].
\label{eq:avg_int_tau}
\end{align}


\subsection{Power control and band selection problem}
\label{sec:problem}
Let $a_t \in \{1,2,...,F\}$ denote the frequency band selected by the SU link and $P_t$ be the transmit power in time slot $t$ . Then, the power control and frequency band selection problem for $t=1,2,...$ can be written as follows:
\begin{align}
\nonumber \{a^*_t, P^*_t\} &= \arg\max_{a_t,P_t} \mathbb{E}\left[R_{a_t,t} \right],
\\ \nonumber  \text{Subject to:~~}& \mathbb{E}_{\tau}\left[ \mathbb{E}[I_{a_t,t} | \tau] \right] \leq I^0,
\\ &P_t \leq P^0,
\end{align}
where $I^0$ is the threshold on the interference towards PU receiver and $P^0$ is the maximum transmit power. By substituting for $R_{f,t}$ and $\mathbb{E}[I_{f,t}|\tau]$ from (\ref{eq:rate}) and (\ref{eq:avg_int_tau}), respectively, we get the following problem statement:
\begin{align}
\nonumber \textbf{(P1)}~~ \nonumber \{a^*_t, P^*_t\} &= \arg\max_{a_t,P_t}\mathbb{E} \left[\log_2\left(1 + { P_t \Gamma_{a_t,t}}\right)\right],
\\ \text{Subject to:~~}& \mathbb{E}_{\tau} \left[P_t M_p (1-\alpha_{a_t}^{2\tau})\right] \leq I^0,
\label{eq:int_cont}
\\ &P_t \leq P^0.
\label{eq:power_cont}
\end{align}

\section{Power control and band selection policies}
\label{sec:policies}
{In this section, we describe two power control schemes: fixed and dynamic. Under the fixed power control scheme, the SU transmits power $P_t = P^{fix}_f$ if it is on frequency band $f$  in slot $t$ and the band is occupied by PUs, i.e., $s_{f,t}\in \{1,2\}$. Under the dynamic power control scheme, the SU transmits power $P_t = P^{dyn}_{f,t}$ if it is on frequency band $f$  in slot $t$ and the band is occupied by PUs, i.e., $s_{f,t}\in \{1,2\}$. The expressions for $P^{fix}_f$ and $P^{dyn}_{f,t}$ are derived in this section.	Note that the SU transmits maximum power $P_t = P^0$ if the PUs are inactive, i.e., $s_{f,t}\in \{0\}$.} In section \ref{sec:band_selection_policies}, we study the performance of frequency band selection policies using the transmit power determined in this section.

\subsection{Power control policies}
\label{sec:power_control_policies}

We first determine the maximum fixed transmit power in band $f$ from SU in order to satisfy constraints (\ref{eq:int_cont}) and (\ref{eq:power_cont}). Let $P_t = P^{fix}_f$ be the the maximum fixed transmit power in band $f$ from SU. The power $P^{fix}_f$ needs to satisfy the following constraints:
\begin{align}
P^{fix}_f M_p \mathbb{E}_\tau \left[1-\alpha_f^{2\tau}\right] \leq I^0,
\\ P^{fix}_f \leq P^0.
\end{align}

Therefore, the maximum fixed transmit power in band $f$ is:
{\small \begin{align}
P^{fix}_f = \min \left( \frac{I^0}{M_p \mathbb{E}_\tau \left[1-\alpha_f^{2\tau}\right]},  P^0  \right)= \left( \frac{I^0}{M_p g(\alpha_f, \mathbf{T}_f)}, P^0 \right) 
\label{eq:p_n_fix_def}
\end{align}}

\noindent where

{\small\begin{multline}
	g(\alpha_f, \mathbf{T}_f) = \frac{\sum \limits_{i} (1-\alpha_f^{2i})}{(\pi_{1,f} + \pi_{2,f})}\times 
	 \bigg(\pi_{1,f} \sum \limits_{s\in \{0,2\}} p_{1s}p_{s2\backslash 1}^{(i-1)} + \\\pi_{2,f} \sum \limits_{s\in \{0,1\}} p_{2s}p_{s1\backslash 2}^{(i-1)} \bigg),
	\label{eq:p_n_fix}
\end{multline}}

\vspace{-1mm}
\noindent and $p_{ss' \backslash s''}^{(i)}$ is the probability of PU link going from state $s$ to state $s'$ in $i$ slots without hitting state $s''$, where $s,s',s'' \in \{0,1,2\}$. $p_{ss'' \backslash s''}^{(i)}$ is obtained from the transition probabilities in $\mathbf{T}_f$. The proof of the last equality is shown in Appendix \ref{app:p_n_fix}.

The SU can also dynamically control the transmit power in each slot to satisfy the constraints (\ref{eq:int_cont}) and (\ref{eq:power_cont}). Let $P^{dyn}_{f,t}$ be the transmit power if SU stays on band $f$ in slot $t$. The interference constraint is satisfied if we have $P^{dyn}_{f,t} M_p (1-\alpha_f^{2\tau}) \leq I^0$. Therefore, the maximum  dynamic power transmitted from the SU is 
\begin{align}
P^{dyn}_{f,t} = \min \left(\frac{I^0}{M_p (1-\alpha_f^{2 \tau})} , P^0\right).
\label{eq:p_n_dyn}
\end{align}

From expressions (\ref{eq:p_n_fix_def}) and (\ref{eq:p_n_dyn}), we can observe that the transmit power increases if $\alpha_f$ is increased. The power also increases of the PU link reversal time $\tau$ is decreased. Thus, the transmit power depends on the temporal correlation as well as the traffic statistic of the PU link. Therefore, the SU can transmit maximum power in the frequency band that has high correlation or small PU link reversal time $\tau$.

Note that the power control is required only when the PU is active in the band selected by SU in slot $t$, i.e., $s_{a_t,t}=\{1,2\}$. The SU transmits maximum power $P_t = P^0$ if $s_{a_t,t}=0$, since there will be no interference to PU receiver in such a slot. 

\subsection{Band selection policies}
\label{sec:band_selection_policies}

Next, we describe four types of polices for band selection using aforementioned power control schemes, namely fixed band fixed power (FBFP), fixed band dynamic power (FBDP), dynamic band fixed power (DBFP), and clairvoyant policy. In the FBFP and FBDP policies, the SU determines which frequency band allows maximum transmit power according to (\ref{eq:p_n_fix}) and (\ref{eq:p_n_dyn}) and stays on that frequency band. On the other hand, in DBFP policies, SU can hop to different frequency bands. The clairvoyant policy assumes a genie SU that can observe all $F$ frequency bands simultaneously in all the slots to maximize its rate.

\subsubsection{Fixed band fixed power (FBFP)}
\label{sec:fbfp}
The simplest band selection and power control policy for the SU is to stay on one frequency band, say $f^*$, and use a fixed transmit power $P^{fix}_{f^*}$. This policy is called fixed band fixed power (FBFP). The band $f^*$ is the frequency band that allows SU to transmit maximum power while satisfying the interference and power constraints. Therefore, we have
\begin{align}
a_t = f^{*} = \arg\max_f P^{fix}_f.
\label{eq:FBFP}
\end{align}

\subsubsection{Fixed band dynamic power (FBDP)}
\label{sec:fbdp}
In this policy as well, the SU stays on one frequency band that allows maximum transmit power. However, the transmit power $P^{dyn}_{f^*,t}$ is dynamically controlled in each slot when the PU is active. The frequency band selected by the SU is same as in the FBFP policy:
\begin{align}
a_t = f^{*} = \arg\max_f P^{fix}_f = \arg\max_f P^{dyn}_{f,t}.
\label{eq:FBDP}
\end{align}

\subsubsection{Dynamic band fixed power (DBFP)}
\label{sec:dbfp}
Under DBFP category, we consider three policies: random, round robin, and DSEE in \cite{liu2013a}. In all three policies, we assume that SU utilizes fixed transmit power $P^{fix}_f$ from (\ref{eq:p_n_fix}) in band $f$. 

In the random policy, the SU selects the frequency band $f$ randomly. Each band has equal probability of getting selected in slot $t$. In the round robin policy, as the name suggests, the SU selects frequency band in round robin fashion. In the DSEE policy presented in \cite{liu2013a}, the band selection problem is treated as a restless multi-armed bandit (MAB) problem. For a fixed transmit power, the problem (\textbf{P1}) can be cast as a MAB problem. Therefore, we apply the DSEE policy in order to dynamically select the frequency band. In the DSEE policy, we consider that if SU is on band $f$ in slot $t$, it uses transmit power $P^{fix}_f$ from (\ref{eq:p_n_fix}) and receives rate (reward) $R_{f,t} = \log_2 \left(1 + P^{fix}_f \Gamma_{f,t} \right)$. The DSEE band selection policy is implemented as described in \cite[Section II.B]{liu2013a}.

\subsubsection{Clairvoyant policy}
\label{sec:clair}
We compare the rate achieved in the aforementioned policies with an ideal,  clairvoyant policy, where a genie-aided SU can observe all $F$ bands simultaneously in each slot $t$, compute null spaces and beamforming vectors in each band and then select the one which provides maximum rate. {In this policy, dynamic power control $P^{dyn}_{f,t}$ is used since it provides higher rate than fixed power, as shown in Theorem \ref{thm:fbdp_v_fbfp} in Section \ref{sec:analysis}}. Therefore, the frequency band selected by the clairvoyant policy in slot $t$, $f_{c,t}$, is given by
\begin{align}
a_t = f_{c,t} &= \arg \max_{f} \log_2 \left(1 + {P^{dyn}_{f,t} \Gamma_{f,t}} \right).
\end{align}
The clairvoyant policy provides an upper bound on achievable rate in the given setting.

\section{Analysis of the policies}
\label{sec:analysis}
In this section, we analyze the achievable rate and interference leakage towards PU under the policies described in the previous section.
\subsubsection{Analysis of FBFP policy}
\label{sec:fbfp_rate}
The expected achievable rate under FBFP policy, if the SU stays on band $f^*$, is denoted by $\mathbb{E}[{R^{(1)}_{f^*,t}}]$. We drop the asterisk in the subscript to simplify the notation. This expression also holds for any frequency band $f \in [1,F]$. The expected rate can be expressed as follows:

{\small
\begin{align}
\nonumber \mathbb{E}[{R^{(1)}_{f,t}}] =& \frac{T_{data}}{T_{slot}} \bigg[\pi_{0,f} \mathbb{E}\left[\log_2\left(1 + {P^0 \Gamma_{f, t}} \right) \vert s_{f,t}=0\right] 
\\&+ (1-\pi_{0,f})  \mathbb{E}\left[\log_2\left(1 + {P^{fix}_{f} \Gamma_{f, t}} \right) \vert s_{f,t}=\{1,2\}\right]\bigg].
\label{eq:fbfp_rate}
\end{align}}

\noindent The first term in above expression is the expected rate when PUs are silent, while the second term is the rate when either PU is active.  As mentioned in Section \ref{sec:Model_problem}, the equivalent channel matrix $\mathbf{H}_{eq,f, t}$ has rank $M_s$ if PU is silent and rank $M_s-M_p$ if one PU is transmitting. Therefore, the expectations in (\ref{eq:fbfp_rate}) can be expressed using the distribution of maximum eigenvalue of rank $M_s$ and $M_s-M_p$ matrices as follows:

{\small \begin{align}
\nonumber &\mathbb{E}\left[\log_2\left(1 + {P^0 \Gamma_{f, t}} \right) \vert s_{f,t}=0\right] 
= \int\limits_{0}^{\infty}\log_2 \left(1 + {P^{fix}_{f} x} \right) f_{M_s}(x)dx
\\\nonumber &\mathbb{E}\left[\log_2\left(1 + {P^{fix}_{f} \Gamma_{f, t}} \right) \vert s_{f,t}=\{1,2\}\right]
 \\&~~~~~~~~~~~~~~~~~~~~~~~~~~ = \int\limits_{0}^{\infty}\log_2 \left(1 + {P^{fix}_{f} x} \right) f_{M_s-M_p}(x)dx,
\label{eq:fbfp_rate2}
\end{align}}

\noindent where $f_{M_s}(x)$ is the probability density function (pdf) of the largest eigenvalue of Hermitian matrix $\mathbf{H}^H_{eq,f,t}\mathbf{H}_{eq,f,t}$ of rank $M_s$. The pdf of the largest eigenvalue is computed using the cumulative distribution function (cdf) $F_{M_s}(x)$ as follows:
\begin{align}
f_{M_s}(x) = \frac{d}{dx}F_{M_s}(x) = \frac{x^{M_s-1}e^{-x}}{\Gamma(M_s)},
\label{eq:pdf_max_eig}
\end{align}
where $F_{M_s}(x) = \frac{\gamma(M_s,x)}{\Gamma(M_s)}$ is the cdf as given in \cite[Eq. 9]{kang2003}, $\gamma(.,.)$ is the incomplete Gamma function and $\Gamma(.)$ is the Gamma function. The expression (\ref{eq:fbfp_rate2}) is computed using the (\ref{eq:pdf_max_eig}) and the expected rate $\mathbb{E}[R^{(1)}_{f,t}]$ is evaluated by substituting (\ref{eq:fbfp_rate2}) in (\ref{eq:fbfp_rate}). Note that the expected interference towards PU remains under FBFP is $I^0$ since the power control scheme ensures that the constraint (\ref{eq:int_cont}) is satisfied with equality.

\subsubsection{Analysis of FBDP policy}
\label{sec:fbdp_rate}
In FBDP, the power is dynamically changed per slot in order to control the interference to the PU. Note that the power $P^{dyn}_{f,t}$ is a function of $\tau$ as shown in (\ref{eq:p_n_dyn}). Therefore, while computing the expected achievable rate, maximum eigenvalue as well as $\tau$ are treated as random variables. The expected rate can be expressed as below:

{\small
\begin{align}
\nonumber \mathbb{E}[R^{(2)}_{f,t}] =&  \frac{T_{data}}{T_{slot}} \bigg[\pi_{0,f} \mathbb{E}\left[\log_2\left(1 + {P^0 \Gamma_{f, t}} \right) \vert s_{f,t}=0\right] 
\\&+ (1-\pi_{0,f})  \mathbb{E}\left[\log_2\left(1 + {P^{dyn}_{f,t} \Gamma_{f, t}} \right) \vert s_{f,t}=\{1,2\}\right]\bigg].
\label{eq:fbdp_rate_v1}
\end{align}}

\noindent The first term is same as the first term in (\ref{eq:fbfp_rate}). The expectation in the second term is computed as follows:
{\small \begin{align}
 \mathbb{E}\left[\log_2\left(1 + {P^{dyn}_{f,t} \Gamma_{f, t}} \right) \right] 
 = \mathbb{E}_\tau \left[ \mathbb{E}_{\Gamma}\left[\log_2\left(1+ {P^{dyn}_{f,t} \Gamma_{f,t}} \right) \vert \tau\right] \right],
\end{align}}

\noindent where $\mathbb{E}_\Gamma$ is the expectation with respect the the maximum eigenvalue assuming $\tau$ is a constant. The condition $s_{f,t}=\{1,2\}$  is dropped from the above expression to simplify the notation. The outer expectation $\mathbb{E}_\tau$ is with respect to $\tau$. The inner expectation is computed by substituting $P^{fix}_f$ with $P^{dyn}_{f,t}$ in (\ref{eq:fbfp_rate2}). Let $z(i) = \mathbb{E}_{\Gamma}\left[\log_2\left(1+ {P^{dyn}_{f,t} \Gamma_{f,t}} \right) \vert \tau=i \right]$ be the inner expectation for $\tau=i$. Then, the outer expectation is given by
\begin{align}
\mathbb{E}_{\tau} [z(i)] = \sum_{i} z(i) \times \Pr(\tau=i),
\label{eq:fbdp_exp_outer}
\end{align}
where $\Pr(\tau=i)$ is computed using 

\begin{align}
\Pr (\tau=i) = \pi_{2,f} \sum_{s\in \{0,1\}} p_{2s}p_{s1\backslash 2}^{(i-1)} + \pi_{1,f} \sum_{s\in \{0,2\}} p_{1s}p_{s2\backslash 1}^{(i-1)}.
\label{eq:pr_tau}
\end{align}
The derivation for the above expression is provided in Appendix \ref{app:p_n_fix}. The expected rate is computed by substituting (\ref{eq:fbdp_exp_outer}) in (\ref{eq:fbdp_rate_v1}). At this point, we state the following theorem comparing the expected rates in FBDP and FBFP:
\begin{theorem}
	The expected rate under fixed band dynamic power (FBDP) exceeds the expected rate under fixed band fixed power (FBFP) policy, i.e., $\mathbb{E}[R^{(2)}_{f,t}]\geq \mathbb{E}[R^{(1)}_{f,t}]$.
	\label{thm:fbdp_v_fbfp}
\end{theorem}
\begin{proof}
Appendix \ref{app:fbdp_v_fbfp}.
\end{proof}

\subsubsection{Analysis of DBFP policies}
\label{sec:dbdp_rate}
In DBFP policies, the SU uses fixed transmit power $P^{fix}_{f}$ when it is on band $f$. In round robin and random band selection policies, each band is selected for same number of times on an average. Therefore, the expected rate under these two policies will be equal. Let $\mathbb{E}[{R^{(3)}_{f,t}}]$ be the expected rate under random and round robin policies. Since each band is visited with equal probability, the expected rate can be written as:
\begin{align}
\mathbb{E}[{R^{(3)}_{f,t}}] =\frac{T_{data}}{T_{slot}} \frac{1}{F}\sum_{f'=1}^{F} \mathbb{E}[R^{(1)}_{f',t}],
\end{align}
where $\mathbb{E}[R^{(1)}_{f',t}]$ is the expected rate under FBFP policy if the SU stays on band $f'$. The expression for $\mathbb{E}[R^{(1)}_{f',t}]$ is obtained from (\ref{eq:fbfp_rate}) by substituting $f=f'$. It can be observed that expected rate under round robin or random band selection is lower than the rate in FBFP, i.e., $\mathbb{E}[{R^{(3)}_{f,t}}]\leq \mathbb{E}[{R^{(1)}_{f,t}}]$. This is because the FBFP policy selects the band that maximizes expected rate, therefore hopping to a different frequency bands does not improve the achievable rate of the SU link. 

For the DSEE policy proposed in \cite{liu2013a}, the performance of DSEE is measured in terms of regret, i.e., difference between the rate received in fixed band policy and the rate received in DSEE. Since the regret is shown to be positive in \cite[Theorem 1]{liu2013a}, we can conclude that DSEE provides lower rate as compared to the fixed band policy (FBFP). Therefore, these DBFP policies using fixed transmit power $P^{fix}_{f}$ do not provide higher rate as compared to the FBFP policy.

The interference towards PU under these policies is higher than the threshold as stated in the following theorem.

\begin{theorem}
	In dynamic band fixed power polices (DBFP), the expected interference leakage towards PU exceeds the threshold $I^0$.
	\label{thm:int_dbfp}
\end{theorem}
\begin{proof}
{Consider that SU follows a DBFP policy, for example, round robin band selection policy and the SU is on band $f$ in slot $t$. Let $s_{f,t}=1$, i.e. PU-1 is the transmitter and PU-2 is the receiver. Under the robin robin policy, SU was on the same band during previous slots $t-F, t-2F, t-3F,\cdots$. The null space to PU receiver (PU-2) was obtained in slot $t-kF$ when PU-2 was the transmitter where $k = \arg \min_{k} \left( s_{f,t-kF}=2 \right)$. Therefore $\tau'=kF$ is the PU link reversal time perceived by the SU under this policy. We can see that $\tau'$ is larger than the PU link reversal time $\tau = \arg \min_{k} (s_{f,t-k}=2)$ in FBFP and FBDP policies wherein the SU stays on the same band and we have $\tau'>\tau$. This holds true for other DBFP policies as well. The expected interference towards PU in band $f$ under DBFP policies is given by $\mathbb{E}_{\tau'}[P^{fix}_f M_p (1-\alpha_f^{2\tau'})]=P^{fix}_fM_p \mathbb{E}_{\tau'}[ (1-\alpha_f^{2\tau'})]$. Since $\tau' > \tau$, we have $\mathbb{E}_{\tau'}[ (1-\alpha_f^{2\tau'})] > \mathbb{E}_{\tau}[ (1-\alpha_f^{2\tau})]$. Further, since the fixed power is given as $P^{fix}_f = I^0/M_p \mathbb{E}_{\tau}[(1-\alpha_f^{2\tau})]$, the average interference in band $f$ under DBFP policies is $P^{fix}_f M_p \mathbb{E}_{\tau'}[ (1-\alpha_f^{2\tau'})] = I^0 \frac{\mathbb{E}_{\tau'}[ (1-\alpha_f^{2\tau'})]}{\mathbb{E}_{\tau}[ (1-\alpha_f^{2\tau})]} > I^0$.}
\end{proof}
\begin{corollary}
	There exists no fixed transmit power for DBFP policies that provides higher rate than the single band policy while satisfying the interference constraint (\ref{eq:int_cont}).
	\label{cor:dbfp}
\end{corollary}
\begin{proof}
Since the fixed power $P^{fix}_{f}$ incurs interference above the threshold $I^0$, one way of satisfying the interference constraint is to transmit lower power $P^{f} < P^{fix}_{f}$ that will satisfy the interference constraint in band $f$. However, this approach  reduces the rate below $\mathbb{E}[R^{(3)}_{f,t}]$ which was already lower than single band policy. Therefore, there is no transmit power that will increase the rate of dynamic band polices while satisfying the interference constraint.
\end{proof}
\subsubsection{Analysis of clairvoyant policy}
\label{sec:clair_rate}
In the clairvoyant policy, we assume that a genie-aided SU observes all $F$ frequency band in each slot, computes the null space to PUs in all bands and then selects the band offering the maximum rate. In this section, we analyze the expected gain of clairvoyant policy over FBFP. By doing so, we can find an upper bound on the achievable rate.

Let us define the expected gain of clairvoyant policy over FBFP as follows:
\begin{align}
\mathbb{E}[g^{(c)}_t] = \mathbb{E} [R^{(c)}_{f,t} - R^{(1)}_{f,t}],
\label{eq:clair_gain}
\end{align}

For simplicity, we consider that clairvoyant policy provides gain in slot $t$ if a) the PU is active in band selected by FBFP, i.e., $s_{f^*,t}=\{1,2\}$ and b) there exists a band $f'$ with no active PU. Under this condition the beamforming gain and transmitted power in band $f'$ will be higher than in band $f^*$. The probability of satisfying this condition in a time slot is $(1-\pi_{0,f^*}) \left(1- \prod_{f' \neq f^*} (1-\pi_{0,f'}\right)$ and the expected gain is given as:

{\small \begin{align}
\nonumber \mathbb{E}[g^{(c)}_t] = \frac{T_{data}}{T_{slot}} (1-\pi_{0,f^*}) \left(1- \prod_{f' \neq f^*} (1-\pi_{0,f'})\right)\\
\times \mathbb{E}\left[\log_2 \left( \frac{1+ P^0 x}{1 + P^{fix}_{f^*} y } \right)\right],
\label{eq:clair_gain2}
\end{align}}

\noindent where $x$ is a random variable with pdf $f_{M_s}(x)$ as mentioned in (\ref{eq:pdf_max_eig}) and  $y$ is a random variable with pdf $f_{M_s-M_p}(y)$. The expectation on the RHS can be written as follows:
\begin{align}
\mathbb{E}\left[\log_2 \left( \frac{1+ P^0 x}{1 + P^{fix}_{f^*} y } \right)\right] =\mathbb{E}_y\left[ \mathbb{E}_x \left[\log_2 \left( \frac{1+ P^0 x}{1 + P^{fix}_{f^*} y } \right) \vert y\right]\right]
\label{eq:clair_gain3}
\end{align}
For a given value of $y$, the inner expectation is a concave function of $x$, it has an upper bound as follows:

{\small \begin{align}
 \nonumber \mathbb{E}_x \left[\log_2 \left( \frac{1+ P^0 x}{1 + P^{fix}_{f^*} y } \right) \vert y \right] \leq \log_2  \left( \frac{1+ P^0 \mathbb{E} [x]}{1 + P^{fix}_{f^*} y } \right) 
\\  = \log_2  \left( \frac{1+ P^0 M_s}{1 + P^{fix}_{f^*} y } \right) 
 \label{eq:clair_gain4}
\end{align}}
 
\noindent Substituting the above inequality in (\ref{eq:clair_gain3}), we get 
{\small \begin{align}
\nonumber \mathbb{E}\left[\log_2 \left( \frac{1+ P^0 x}{1 + P^{fix}_{f^*} y } \right)\right] &\leq \mathbb{E}_y\left[ \log_2 \left( \frac{1+ P^0 M_s}{1 + P^{fix}_{f^*} y } \right) \right]
\\   = \frac{1}{\Gamma(M_s-M_p)}&\int\limits_{0}^{\infty} \log_2 \left( \frac{1+ P^0 M_s}{1 + P^{fix}_{f^*} y } \right) y^{M_s - M_p -1} e^{-y}dy.
\label{eq:clair_gain5}
\end{align}}

\noindent Therefore, the expected gain $ \mathbb{E}[g^{(c)}_t]$ in (\ref{eq:clair_gain2}) us upper bounded as follows:
{\small \begin{align}
\nonumber \mathbb{E}[g^{(c)}_t] \leq \frac{T_{data}}{T_{slot}} \frac{(1-\pi_{0,f^*}) \left(1- \prod \limits_{f' \neq f^*} (1-\pi_{0,f'})\right)}{\Gamma(M_s-M_p)}  
  \\ \times \int\limits_{0}^{\infty} \log_2 \left( \frac{1+ P^0 M_s}{1 + P^{fix}_{f^*} y } \right) y^{M_s - M_p -1} e^{-y}dy =  g^{c}_{max}
  \label{eq:clair_gain6}
\end{align}}

Note that, in the above equation, $g^{c}_{max}$ depends on the temporal correlation $\alpha_{f}$ through $P^{fix}_{f}$. The transition probabilities of PU links also affect the gain through $P^{fix}_{f^*} = \min \left(\frac{I^0}{M_p g(\alpha_{f^*}, \mathbf{T}_{f^*})} , P^0 \right)$. We can see that as $\alpha_{f^*} \rightarrow 1$, the power $P^{fix}_{f^*} \rightarrow P^0$ and the difference between $P^0 M_s$ and $P^{fix}_{f^*}y$ reduces for any given $y$. Therefore, higher temporal correlation decreases the gain of clairvoyant policy. It should be noted that the clairvoyant policy provides the maximum rate amongst all possible policies. If the gain of this policy over FBFP is small, then it means that there exists no policy that achieves significantly higher rate than the fixed band policy.


\begin{table}
	\centering
	\caption{\small PU traffic configurations}
	\vspace{0mm}
		\begin{tabular}{|c|c|c|}
			\hline
			Traffic config. & Transition probability matrix & $\mathbb{E}[\tau]$\\
			\hline \hline
			0 & $\mathbf{T}_f = \begin{bmatrix}
			0 & 0 & 1	\\
			1 & 0 & 0	\\
			0 & 0.2 & 0.8\\
			\end{bmatrix}$ 	 & 4.43\\ \hline
			1 & $\mathbf{T}_f = \begin{bmatrix}
			0 & 0 & 1 	\\
			0.67 & 0.33 & 0	\\
			0 & 0.5  &0.5 \\
			\end{bmatrix}$ 	 & 1.83\\ \hline
			2 & $\mathbf{T}_f =  \begin{bmatrix}
			0 & 0 & 1\\
			0.4 & 0.6 & 0\\
			0 & 1 & 0\\
			\end{bmatrix}$ 	 & 1.83\\ \hline
			3 & $\mathbf{T}_f = \begin{bmatrix}
			0 & 0 & 1\\
			0.2 &0.8 &0\\
			0 & 0.33 & 0.67\\
			\end{bmatrix}$ 	 & 4.11\\ \hline
			4 & $\mathbf{T}_f = \begin{bmatrix}
			0 & 0 & 1\\
			0.17 & 0.83 & 0\\
			0 & 0.5 & 0.5\\
			\end{bmatrix}$ 	 & 4.67\\ \hline
			5 & $\mathbf{T}_f = \begin{bmatrix}
			0 & 0 & 1\\
			0.14 & 0.86 &	0 \\
			0 & 1 & 0\\
			\end{bmatrix}$ 	 & 5.67 \\ \hline
			6 & $\mathbf{T}_f = \begin{bmatrix}
			0 & 0 & 1\\
			1 & 0 & 0\\
			0 & 0.4 & 0.6 \\
			\end{bmatrix}$ 	 & 2.17 \\ \hline	
		\end{tabular}	
	\label{tab:traffic_config}
	\vspace{-5mm}	
\end{table}

{
	\subsection{Overhead and computational complexity}
	\label{sec:complexity}
	In order to implement the policies, the SU pair  requires the knowledge of transition probability matrix $\mathbf{T}_f$ and temporal correlation $\alpha_f$. Acquisition of $\mathbf{T}_f$ and $\alpha_f$ results in additional overhead and complexity in implementation of the policies as discussed below.
	\subsubsection{Computation of transition probabilities}
	\label{sec:trans_prob}
	The transition probabilities for PU link, $p_{kl,f}$, are computed from the knowledge of the traffic configuration used by PU. We assume that PUs follow LTE TDD traffic configurations  described in  {3GPP 36.211} \cite{3gpp2017} specifying which slots are used for uplink and downlink in one LTE subframe. Without the loss of generality, we can assume that transmission from PU-1 to PU-2 is downlink ($s_{f,t}=1$) and transmission from PU-2 to PU-1 is uplink ($s_{f,t}=2$). 
	
	The traffic configuration, indicated by an integer between 0 and 6, is set by the operator of the PU network and can be conveyed to SUs. Using the knowledge of traffic configuration, the SU can compute the number of state transitions of the PU link in one LTE subframe. The transition probabilities, $p_{kl,f}$, are computed by counting the number of transitions in PU link state $k$ to $l, k,l\in \{0,1,2\}$ in a subframe and diving by the total number of slots in the subframe. The transition probabilities for traffic configuration 0 to 6 are shown in Table \ref{tab:traffic_config}.

	The configuration remains unchanged for a long duration, usually hours \cite{3gpp2012}. The overhead of obtaining $\mathbf{T}_f$ depends on how often the configuration changes. The PU operator needs to provide the information only when the configuration is changed. The matrix $\mathbf{T}_f$ for each traffic configuration can be stored in a memory at SU. Thus, $\mathbf{T}_f$ is deterministically obtained without error from the knowledge of the traffic configuration and requires no additional runtime computations. 

	\subsubsection{Computation of temporal correlation}
	\label{sec:alpha_f}
	The temporal correlation coefficient $\alpha_f$ can be estimated using the covariance matrices $\mathbf{\hat{Q}}_{i,f,t}, i=1,2$ computed during $T_{sense}$ duration. The estimation algorithm proposed in \cite{wen2006} can be used to compute $\alpha_f$. The associated complexity is, $\mathcal{O}(M_s^2)$, same as that of computing the covariance matrix. 
}

\section{Simulation Results}
\label{sec:results}
In this section, we compare the performance of the policies in terms of achievable rate at SU and interference towards PU. For the simulations, the matrix $\mathbf{T}_{f}$ for PU link is constructed using traffic models of TDD LTE in {3GPP 36.211} \cite{3gpp2017} considering PU-1 is the base station and PU-2 is UE. Therefore, downlink is treated as $s_{f,t}=1$ and uplink is $s_{f,t}=2$. The transition probabilities $p_{kl,f} = \Pr(s_{f,t+1}=l|s_{f,t}=k), k,l \in \{0,1,2\}$ under these configurations are computed and are shown in Table \ref{tab:traffic_config} along with the average time for link reversal $\mathbb{E}[\tau]$ under the model. It is computed using the following expression:
\begin{align}
\mathbb{E}[\tau] = \sum_{i} i  \times \Pr (\tau=i),
\end{align}
where $\Pr (\tau=i)$ is obtained from (\ref{eq:pr_tau}). The temporal fading coefficient is modeled as $\alpha_f = J_0(2\pi f_d T_{slot})$, where $J_0(.)$ is the 0-th order Bessel function, $f_d$ is the Doppler frequency and $T_{slot}= 1$ ms. The fraction of time slot used for SU data transmission is $T_{data}/T_{slot} = 0.8$, while that for sensing is  $T_{sense}/T_{slot} = 0.2$ \cite{gabran2011}. The number of antennas at the SUs and PUs are $M_s=4$ and $M_p=1$, respectively, unless specified otherwise. Total transmit power and interference thresholds are $P^0/\sigma^2_w = 20$dB and $I^0/\sigma^2_w = -10$dB. Analytical and simulation results are shown for the power control and band selection policies.

\begin{figure}[t!]
	\centering		
	\begin{subfigure}[b]{\columnwidth}
		\centering
		\includegraphics[width=0.8\columnwidth]{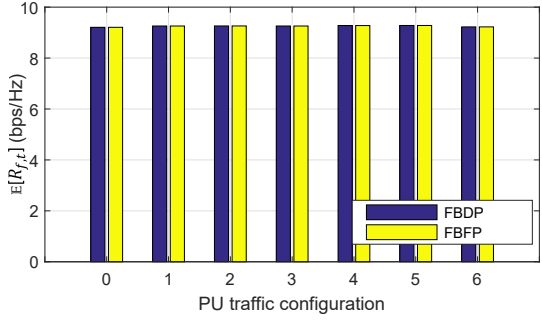}
		\vspace{-1mm}
		\caption{\small $\alpha=0.9998$ ($f_d = 5$Hz)}
		\label{fig:fbfp_vs_fbdp_large_alpha}
	\end{subfigure}
	\begin{subfigure}[b]{\columnwidth}
		\centering
		\includegraphics[width=0.8\columnwidth]{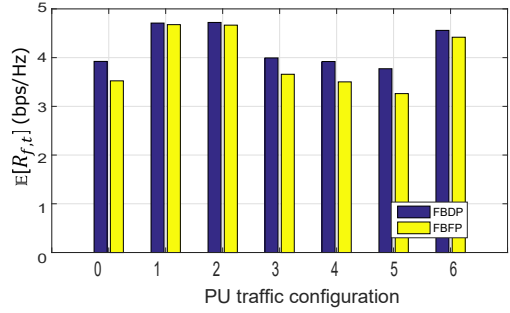}
		\vspace{-1mm}
		\caption{\small $\alpha=0.9938$ ($f_d = 25$Hz)}
		\label{fig:fbfp_vs_fbdp_small_alpha}
	\end{subfigure}		
	\vspace{-2mm}
	\caption{\small Comparison between the achievable rate of SU under FBFP and FBDP with different PU traffic configurations. $F=1$. $I^0/\sigma^2_w=-10$dB. $P^0/\sigma^2_w=20$dB.}
	\label{fig:fbfp_vs_fbdp_traffic_model}
	\vspace{-5mm}
\end{figure}

\subsubsection{Comparison between FBFP and FBDP for $F=1$}
\label{sec:results_single_band}
First, we compare the performance of fixed band policies: FBFP and FBDP for $F=1$ under the PU traffic models described above. The average rate under the two polices is shown in Fig. \ref{fig:fbfp_vs_fbdp_traffic_model} for $\alpha_f = 0.9998$ and $\alpha_f = 0.9938$. We observe that both the policies provide same rate when  $\alpha_f = 0.9998$. For $\alpha_f \rightarrow 1$, the dynamic power $P^{dyn}_{f,t} \rightarrow P^0$ and it does not change significantly with $\tau$, hence it becomes approximately constant as in FBFP. Therefore, the two policies provide same rate as shown in Fig.\ref{fig:fbfp_vs_fbdp_large_alpha}. On the other hand, as the temporal correlation decreases to $\alpha_f = 0.9938$ as shown in Fig. \ref{fig:fbfp_vs_fbdp_small_alpha}, the rate achievable rate differs under the two policies. The rate is maximum under PU traffic models 1 and 2, while its smallest under PU traffic model 5. It can be observed that the achievable rate in this case varies inversely with the average link reversal time $\mathbb{E}[\tau]$ shown in Table \ref{tab:traffic_config}. For smaller link reversal time, the PU switches its role from transmitter to receiver in short duration. Therefore, the SU transmitter has more accurate null space to the PU receiver in a given time slot and it can transmit higher power while still keeping the interference below the threshold. This in turn results in higher rate for the SU. Further, it can be seen that the rate under the two polices reduces with smaller temporal correlation as shown in Fig. \ref{fig:fbfp_vd_fbdp_alpha}. This is due to the fact for larger $\mathbb{E}[\tau]$ or smaller $\alpha_f$ the SU transmitter needs to transmit lower power $P^{fix}_{f}$ and $P^{dyn}_{f,t}$ according to (\ref{eq:p_n_fix}) and (\ref{eq:p_n_dyn}), respectively, which in turn reduces the achievable rate. These results also confirm Theorem \ref{thm:fbdp_v_fbfp} since rate under FBDP is no smaller than in FBFP.

\begin{figure}[t!]
	\centering		
	\includegraphics[width=0.8\columnwidth]{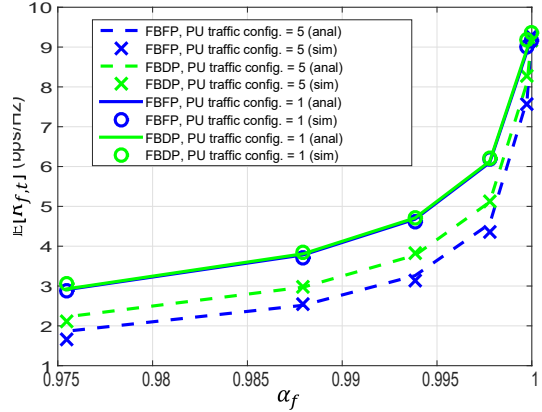}
	\caption{\small Rate at SU under FBFP and FBDP for different temporal correlations $\alpha_f \in [0.9755, 0.9998]$ or Doppler rate $f_d \in [5, 50]$. $F=1$. $I^0/\sigma^2_w=-10$dB. $P^0/\sigma^2_w=20$dB.}	
	\label{fig:fbfp_vd_fbdp_alpha}
	\vspace{-2mm}
\end{figure}

It can also be observed from Fig. \ref{fig:fbfp_vs_fbdp_small_alpha} and \ref{fig:fbfp_vd_fbdp_alpha} that the difference in the rate of SU link under FBFP and FBDP is negligible if PU traffic configuration is 1 and 2, i.e., when $T_{f,rev}$ is small. As the link reversal time approaches 1 as in the case of traffic configurations 1 and 2, the transmitted power  under the two polices become similar: $P^{dyn}_{f,t} \approx P^{fix}_{f} \rightarrow \frac{I^0}{M_p(1-\alpha_f^{2\tau})}$, resulting in similar achievable rates.

The fixed band polices FBFP and FBDP select the band that maximizes the power $P^{fix}_{f}$, as this band maximizes the expected rate. The transmitted power is inversely proportional to the average link reversal time of the traffic models as shown in Fig. \ref{fig:P_fix_v_T_rev}. We can observe that if there are multiple bands available with same  temporal correlations $\alpha_f$, then the fixed band policies select the band with lowest $\mathbb{E}[\tau]$. Similarly, if the link reversal time is same in different bands, the policies would select the band with maximum temporal correlation.

\begin{figure}[t!]
	\centering		
	\includegraphics[width=0.9\columnwidth]{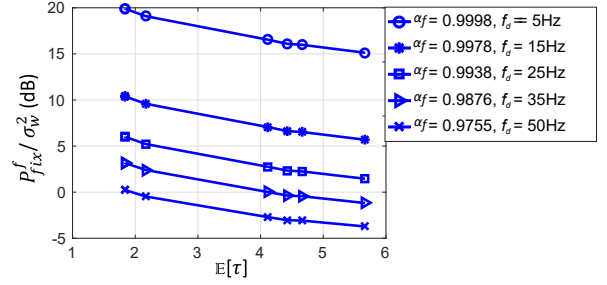}
	\vspace{-2mm}
	\caption{\small Fixed transmit power as a function of $\alpha_f$ and PU link reversal time $\mathbb{E}[\tau]$.}
	\label{fig:P_fix_v_T_rev}
	\vspace{-4mm}
\end{figure}

\begin{figure}[t!]
	\centering		
	\begin{subfigure}[b]{\columnwidth}
		\centering
		\includegraphics[width=0.9\columnwidth]{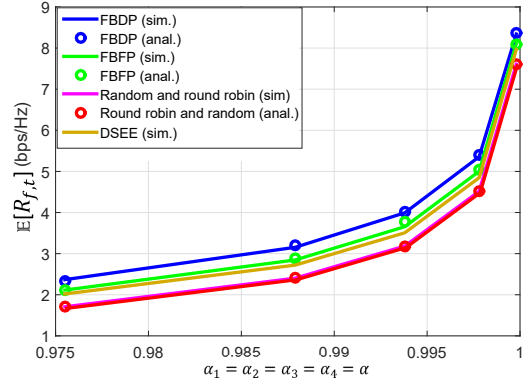}
		\caption{\small Rate under different policies.}
		\label{fig:all_policies_rate}
	\end{subfigure}
	\begin{subfigure}[b]{\columnwidth}
		\centering
		\includegraphics[width=0.9\columnwidth]{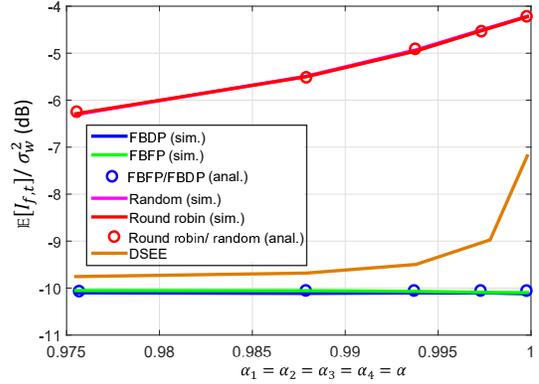}
		\caption{\small Interference towards PU under different policies.}
		\label{fig:all_policies_int}
	\end{subfigure}		
	\caption{\small Comparison between policies with $F=4$ bands. PUs in band 1, 2, 3, and 4 follow traffic configurations 0, 3, 4, and 5, respectively. $I^0/\sigma^2_w=-10$dB. $P^0/\sigma^2_w=20$dB.}
	\label{fig:all_policies}
	\vspace{-4mm}
\end{figure}

\subsubsection{Comparison of fixed and dynamic band policies for $F=4$}
\label{sec:results_multiband}
In this section, we compare FBFP and FBDP with DBFP policies when $F=4$ and bands 1, 2, 3, and 4 have PU traffic configurations $0,3,4$ and $5$, respectively. We selected the 4 traffic configurations with smallest $\mathbb{E}[\tau]$ so that powers  $P^{dyn}_{f,t}$ and $P^{fix}_f$ are not similar. The achievable rate and interference towards PU is under different policies is shown in Fig. \ref{fig:all_policies}. We can see that the FBDP policy provides higher rate than other policies as shown in Fig. \ref{fig:all_policies_rate}. Further, the interference under dynamic band policies is higher than the required threshold as mentioned in Theorem \ref{thm:int_dbfp}. It is interesting to note that the interference under round robin, random and DSEE is not only higher than FBDP and FBFP, but it also increases with increased temporal correlation. This counter-intuitive observation can be explained as follows. 

For a given traffic configuration, the function $g(\alpha_f, \mathbf{T}_f)$ in (\ref{eq:p_n_fix}) depends only on the temporal correlation $\alpha_f$. As the temporal correlation increases $g(\alpha_f, \mathbf{T}_f)$ reduces and higher power $P^{fix}_f$ is transmitted by the SU. While computing the transmit power $P^{fix}_f$, the underlying assumption is that the SU-1 has the latest null space to PU receiver. However, since SU is hopping to different bands in multi-band policies, it has older null space than what is assumed in the computation. This in turn increases the interference towards PU as shown in Fig. \ref{fig:all_policies_int}. The interference under DSEE is lower than in random and round robin due to the fact that the SU stays on the same band for a longer time under DSEE before exploring other bands \cite{liu2013a}. On the other hand, SU hops to different bands more frequently under round robin and random policies. This results in older null spaces at SU-1 causing significantly high interference towards PU receiver.

Note that to limit the interference below $I^0$, DBFP policies need to reduce the transmit power, which would further reduce the rate of the SU link. Therefore, the SU cannot simultaneously contain the interference and provide higher rate in DBFP policies as compared to fixed band policies as mentioned in Corollary \ref{cor:dbfp}.

\begin{figure}[t!]
	\centering		
	\includegraphics[width=0.8\columnwidth]{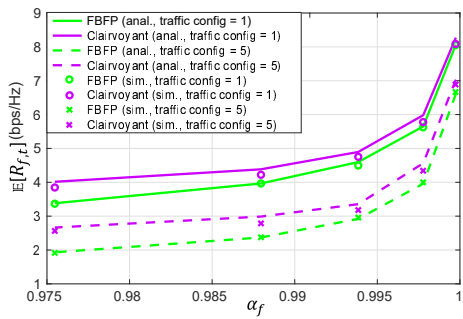}
	\caption{\small Impact of temporal correlation on the rate under FBFP and clairvoyant policies with $F=4$ bands each following same traffic configuration. $I^0/\sigma^2_w=-10$dB. $P^0/\sigma^2_w=20$dB. $M_s=4$.}
	\label{fig:clair_fbfp_alpha}
	\vspace{-2mm}
\end{figure}
\begin{figure}[t!]
	\centering		
	\includegraphics[width=0.8\columnwidth]{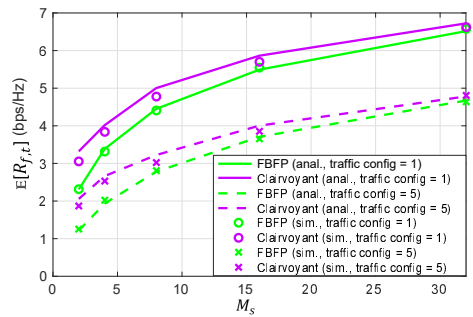}
	\caption{\small Impact of number of SU antennas on the rate under FBFP and clairvoyant policies with $F=4$ bands each following same traffic configuration. $I^0/\sigma^2_w=-10$dB. $P^0/\sigma^2_w=20$dB. $\alpha_f = 0.9755$ in each band ($f_d=50$Hz).}
	\label{fig:clair_fbfp_Ms}
	\vspace{-5mm}
\end{figure}

\subsubsection{Comparison with clairvoyant policy}
\label{sec:results_fbdp_clairvoyant}
In this section, we compare the achievable rate of clairvoyant policy and FBFP assuming all $F=4$ channels follow the same traffic configuration and temporal correlations. We consider two traffic configurations $1$ and $5$ with maximum and minimum link reversal time, respectively. It can be observed in Fig. \ref{fig:clair_fbfp_alpha} that the gain of clairvoyant policy reduces as temporal correlation increases as explained in Section \ref{sec:clair_rate}. The impact of increasing the number antennas at SUs is shown in Fig. \ref{fig:clair_fbfp_Ms}. The gain of clairvoyant policy reduces with increased number of antennas. In other words, the SU does not loose significant amount of rate even if it stays in one band. Therefore, finding empty time slots in other bands in clairvoyant policy is less beneficial in terms of increasing the rate of the SU. From these observations, we conclude that the gap between clairvoyant policy and the fixed band policy can be reduced by increasing the number of antennas at SUs.


\begin{figure}
	\centering
	\includegraphics[width=0.8\columnwidth]{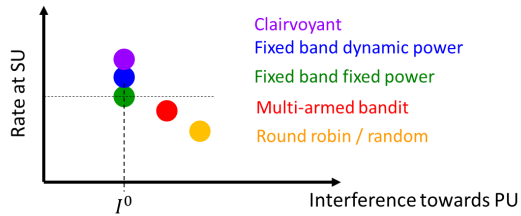}
	\caption{\small {Relative performance of policies studied in this paper.}}
	\label{fig:policy_fig_last}
	\vspace{-2mm}
\end{figure}

\section{Conclusion}
\label{sec:Conclusion}
We studied power control and frequency band selection policies for multi-band underlay MIMO cognitive radio with the objective of maximizing the rate of the SU while keeping the interference leakage towards PUs below specified level. First, we derived expressions for transmit power from SU for fixed and dynamic power control schemes. Then, the we studied the performance of band selection policies that use the proposed transmit power schemes. In the fixed band policies, we proved that the FBDP policy provides higher rate than the FBFP policy, while both policies keep the interference towards PU to the specified threshold. We showed that the DBFP policies, such as round robin, random and the DSEE policy based on multi-armed bandit framework, result in higher interference to PUs as compared to fixed band policies. In conclusion, the relative performance of the policies can be represented as shown in Fig. \ref{fig:policy_fig_last} which shows rate at SU versus interference to PU under the policies studied in this paper. As expected, the genie-aided clairvoyant policy provides maximum rate at the SU. We have provided an expression for the gap between the rate achieved in optimal clairvoyant policy and the FBFP policy. We show that the gap is reduced under slow-varying channels or as the number of SU antennas is increased. 

{
	\section {Future extension}
	\label{sec:extension}
	The band selection policies can be extended to a more general CR network where $N$ SU pairs are co-ordinated by one central node that allocates frequency bands to SU pairs. The central node can allocate bands in order to maximize the sum rate of SUs. Let us consider that SU pairs are indexed by $n$. In order to implement the fixed band policy, the central node computes the transmit power and corresponding rate $R_{n,f}$ for each SU pair $n$ and frequency band $f$ using (\ref{eq:fbfp_rate}). Then, it assigns frequency bands to the SU pairs by solving a weighted bipartite matching problem with the objective of maximizing the sum rate $\sum_{n,f}I_{n,f}R_{n,f}$, where $I_{n,f}$ is a binary assignment variable. Once the optimal assignment is determined, the SU pairs can stay on the assigned band and utilize the proposed power control scheme to transmit their signals while limiting the interference to PUs. The problem of allocating frequency bands to multiple SU pairs can also be formulated as a resource allocation problem by extending the work in \cite{marques2012}.
}

\appendices

\section{Estimation of PU link state}
\label{app:st_est}
{SU-1 and SU-2 can independently estimate $s_{f,t}$ based on the received signals during the sensing duration. Here, we describe estimation at SU-1. Whether the PU link is active, i.e. $s_{f,t}\in \{1,2\}$ or inactive, i.e. $s_{f,t}=0$, is identified by energy detection \cite{soltanmohammadi2013}. In order to identify $s_{f,t}=1$ from $s_{f,t}=2$, the signal received at SU-1 are classified using a hypothesis test. Consider that SU-1 computes the received covariance $\mathbf{\hat{Q}}_{1,f,t-\tau}$ and null space $\mathbf{A}_{1,f,t-\tau}$ in slot $t-\tau$ and labels\footnote{The labels 1 and 2 of the PU state can be reversed without affecting the operation.} the PU state as $s_{f,t-\tau}=1$. Then in slot $t$, SU-1 computes the covariance matrix $\mathbf{\hat{Q}}_{1,f,t}$. In order to determine whether $s_{f,t}=1$ or $s_{f,t}=2$, SU-1 runs a binary hypothesis test: $\mathcal{H}_1$ indicates $s_{f,t}=1$ and $\mathcal{H}_2$ indicates  $s_{f,t}=2$. For the test, SU-1 uses the signal power received in the previously computed null space $P_{null} = \text{Tr}(\mathbf{A}_{1,f,t-\tau}^H \mathbf{\hat{Q}}_{1,f,t}\mathbf{A}_{1,f,t-\tau})$ \cite{chaudhari2017}. The hypothesis test is described as:

{\small \begin{align}
	P_{null} =  \text{Tr}(\mathbf{A}_{1,f,t-\tau}^H \mathbf{\hat{Q}}_{1,f,t}\mathbf{A}_{1,f,t-\tau}) \mathop{\gtreqless}_{\mathcal{H}_1}^{\mathcal{H}_2} P_{th},
	\label{eq:hypothesis}
\end{align}}

\vspace{-4mm}
\noindent where $P_{null}$ is the component of estimated received power $\text{Tr}(\mathbf{\hat{Q}}_{1,f,t})$ in the subspace spanned by columns of $\mathbf{A}_{1,f,t-\tau}$. Under $\mathcal{H}_1$, the asymptotic estimate of the received power is given as 

{\small \begin{align}
\nonumber \text{Tr}(\mathbf{{Q}}_{1,f,t})& = \text{Tr}\left(\mathbf{G}_{11,f,t} \mathbb{E}\left[\mathbf{x}_1(n)\mathbf{x}^H_1(n)\right]\mathbf{G}^H_{11,f,t} +\sigma^2_w \mathbf{I}\right) \\&= P_{1,x} \text{Tr}\left(\mathbf{G}_{11,f,t} \mathbf{G}^H_{11,f,t}\right) +M_s \sigma^2_w,
\end{align}}

\vspace{-2mm}
\noindent
where $P_{1,x}$ is the transmit power from PU-1 and $\mathbb{E}\left[\mathbf{x}_1(n)\mathbf{x}^H_1(n)\right] = P_{1,x} \mathbf{I}$. However, since SU has only non-asymptotic estimate, the estimated received power $\text{Tr}(\mathbf{\hat{Q}}_{1,f,t})$ is modeled as a Gaussian random variable with mean $\mu = \text{Tr}(\mathbf{{Q}}_{1,f,t}) = M_s\sigma^2_w (\text{SNR}  + 1)$ and variance $\sigma^2 = \frac{1}{N}\left(\mu^2 + \sigma^4_w\right)$, where $\text{SNR} = \frac{ P_{1,x} \text{Tr}\left(\mathbf{G}_{11,f,t} \mathbf{G}^H_{11,f,t}\right)}{M_s \sigma^2_w}$ \cite[Eq. (12)]{laghate2017}. Further, using the Gauss-Markov model

{\small \begin{align}
\mathbf{{G}}_{11,f,t} = \alpha_f^\tau \mathbf{{G}}_{11,f,t-\tau} + \sqrt{1-\alpha_f^2} \sum_{\tau' = 0}^{\tau-1} \alpha_f^{\tau - \tau' - 1} \Delta \mathbf{{G}}_{11,f,t-\tau'}
\end{align}}

\vspace{-2mm}
\noindent and the definition of correlation in channel vectors in \cite[Eq. (8)]{choi2005}, we can see that the columns of $\mathbf{{G}}_{11,f,t}$ and $\mathbf{{G}}_{11,f,t-\tau}$ are correlated with correlation $\alpha_f^{2\tau}$. As shown in \cite[Fig. 2]{choi2005}, the correlation in channel vectors results in the same correlation in eigenvectors if $\alpha_f^{2\tau}\geq 0.7$. Therefore, we get

{\small\begin{align}
\nonumber \frac{\mathbb{E}[|\mathbf{{G}}_{11,f,t}(:,k)\mathbf{{G}}^H_{11,f,t-\tau}(:,k)|^2]}{\mathbb{E}\left[||\mathbf{{G}}_{11,f,t}(:,k)||^2\right] \mathbb{E}\left[||\mathbf{{G}}_{11,f,t-\tau}(:,k)||^2\right]} ~~~~~~~~~~~~~~~~~~~\\~~~~~~~~~~~~= \mathbb{E}[|\mathbf{{A}}_{1,f,t}(:,l)\mathbf{{A}}^H_{1,f,t-\tau}(:,l)|^2] =\alpha_f^{2\tau}, 
\end{align}}

\vspace{-1mm}
\noindent where $k \in \{1,2\cdots,M_p\}, l\in\{1,2,\cdots,M_s-M_p\}$, and $\mathbf{A}(:,l)$ is the $l$th column of matrix $\mathbf{A}$. Further, using the orthogonality between $\mathbf{{A}}_{1,f,t-\tau}$ and $\mathbf{{G}}_{1,f,t-\tau}$ and the Guass-Markov model, we can write the correlation between $\mathbf{{A}}_{1,f,t-\tau}$ and $\mathbf{{G}}_{1,f,t}$ as follows:

{\small\begin{align}
\nonumber \mathbb{E}[|\mathbf{{G}}_{11,f,t}(:,k)\mathbf{{A}}^H_{1,f,t-\tau}(:,l)|^2] &=  (1-\alpha_f^{2\tau})\mathbb{E}\left[||\mathbf{{G}}_{11,f,t}(:,k)||^2\right] 
\end{align}}

\vspace{-2mm}
\noindent Therefore, the mean of $P_{null}$ under $\mathcal{H}_1$ is as follows:
{\begin{align}
\nonumber \mu_{P }&=\mathbb{E}[P_{null}\vert \mathcal{H}_1] = \text{Tr}(\mathbf{A}_{1,f,t-\tau}^H \mathbf{Q}_{1,f,t}\mathbf{A}_{1,f,t-\tau}) 
\\\nonumber &= (1-\alpha_f^{2\tau})P_{1,x}\text{Tr}\left(\mathbf{G}_{11,f,t} \mathbf{G}^H_{11,f,t}\right) + M_s \sigma^2_w \\
&= (1-\alpha_f^{2\tau})\mu + \alpha_f^{2\tau}M_s \sigma^2_w.
\label{eq:mean_P_null_H1}
\end{align}}

\vspace{-1mm}
\noindent  The variance due to non-asymptotic estimation is expressed as in \cite[Eq.8]{chaudhari2017}:

{\small \begin{align}
\sigma^2_{P} = \text{Var}({P_{null}}\vert \mathcal{H}_1) = (1-\alpha_f^{2\tau})^2 \sigma^2 = \frac{(1-\alpha_f^{2\tau})^2}{N}(\mu^2 + \sigma^4_w).
\label{eq:var_P_null_H1}
\end{align}}

\vspace{-1mm}
The probability of miss-classifying $s_{f,t}=2$ when $s_{f,t}=1$ is $p_{m}=Q\left(\frac{P_{th}- \mu_P}{\sigma_P}\right)$, where $Q(.)$ is the Q-function. For a fixed $p_m$, the threshold can be set as $P_{th} = Q^{-1}(p_m)\sigma_P + \mu_P$ and the estimate of $\mu_P$ and $\sigma_P$ are computed by replacing $\mu= \text{Tr}(\mathbf{Q}_{1,f,t})$ in (\ref{eq:mean_P_null_H1}) with $\text{Tr}(\mathbf{\hat{Q}}_{1,f,t})$. 

{Under $\mathcal{H}_2$, we have $\mathbf{Q}_{1,f,t} = P_{2,x} \mathbf{G}_{21,f,t}\mathbf{G}^H_{21,f,t} + \sigma^2_w  \mathbf{I}$, where $P_{2,x}$ is the transmit power from PU-2 and SNR is $\frac{P_{2,x} \text{Tr}\left(\mathbf{G}_{21,f,t} \mathbf{G}^H_{21,f,t}\right)}{{M_s \sigma^2_w}}$.	Since the columns of $\mathbf{A}_{1,f,t-\tau}$ and $\mathbf{G}_{21,f,t}$ are uncorrelated, $P_{null}$ is a gamma random variable with shape parameter $\kappa$ and scale parameter $\theta$ given as follows:

{\small \begin{align}
\kappa = \frac{M_s M_p \left(P_{2,x} + \sigma^2_w +\sigma^2_w/N \right)^2}{P^2_{2,x} + \sigma^4_w +\sigma^4_w/N^2},~
\theta = \frac{P^2_{2,x} + \sigma^4_w +\sigma^4_w/N^2}{P_{2,x} + \sigma^2_w +\sigma^2_w/N},
\label{eq:P_null_H2}
\end{align}}

\vspace{-1mm}
\noindent The above follows from \cite[Lemma 2 and 3]{hosseini2014}. Thus, for a fixed threshold $P_{th}$, the probability of error in state estimation is given as 	

{\small \begin{align}
\nonumber p_e &=  \pi_{1,f} \Pr\left(P_{null} > P_{th} \vert \mathcal{H}_1\right) + \pi_{2,f} \Pr\left(P_{null} \leq P_{th} \vert \mathcal{H}_2\right )\\
		&= \pi_{1,f} Q\left(\frac{P_{th} - \mu_{P}}{\sigma_{P }}\right) + \pi_{2,f} \mathcal{F}(\kappa,\theta, P_{th}),
		\label{eq:state_est_error}
\end{align}}

\noindent where $\mathcal{F}(\kappa,\theta, P_{th}) = \gamma\left(\kappa, \frac{P_{th}}{\theta}\right) / \Gamma(\kappa)$ is the CDF of $P_{null}$  under $\mathcal{H}_2$.
	
If there is an error in the state estimation, SU-1 utilizes incorrect null space for precoding, which results in higher interference leakage to PU receiver. For example, if $s_{f,t}=1$ and estimated state is $s_{f,t}=2$, then the expected interference to PU receiver is $\mathbb{E}[P_t ||\mathbf{G}_{12,f,t}^H \mathbf{v}_t||^2]=P_t M_p$. Therefore, the expected interference to PU receiver considering the state estimation error is

\begin{align}
\mathbb{E}[I_{f,t}] = (1-p_e) P_t M_p (1-\alpha_f^{2\tau}) + p_e P_t M_p,
\label{eq:avg_int_v2}
\end{align}

where the first term follows from the discussion in Section \ref{sec:interference_to_pu}. For $p_e \ll \frac{1-\alpha_f^{2\tau}}{2-\alpha_f^{2\tau}}$, we can ignore the impact of error in state estimation and approximate (\ref{eq:avg_int_v2}) with (\ref{eq:avg_int}). For example, $p_e \leq  0.01\left( \frac{1-\alpha_f^{2\tau}}{2-\alpha_f^{2\tau}} \right)$ holds for the 7 traffic configurations of PU link considered in this paper under SNR $\geq 3$dB and $\tau \leq 10$ at $M_s=4$, $M_p=1$, $\alpha_f\in [0.9755,0.9999]$ if the threshold is set as $P_{th} = Q^{-1}(10^{-4})\mu_P + \sigma_P$ for $p_m=10^{-4}$. Further, increasing $M_s$ and $M_p$ reduces $p_e$.  This is due to the fact that the shape parameter $\kappa$ of $P_{null}$ under $\mathcal{H}_2$ is proportional to $M_s M_p$ as shown in  (\ref{eq:P_null_H2}). Increasing $M_s$ or $M_p$ increases the mean of $P_{null}$ under $\mathcal{H}_2$ which results in lower probability of error.


\section{Proof: $\mathbb{E}[P_t ||\mathbf{G}^H_{11,f,t}\mathbf{v}_t||^2] = P_t M_p (1-\alpha_f^{2\tau})$}
\label{app:avg_int}
Using the Gauss-Markov model, the relationship between $\mathbf{G}^H_{11,f,t}$ and $\mathbf{G}^H_{11,f,t-\tau}$ can be expressed as follows:

{\small\begin{align}
\mathbf{G}^H_{11,f,t} = \alpha_f^\tau \mathbf{G}^H_{11,f,t-\tau} + \sqrt{1-\alpha_f^2} \sum_{\tau' = 0}^{\tau-1} \alpha_f^{\tau - \tau' - 1} \Delta \mathbf{G}^H_{11,f,t-\tau'}.
\end{align}}

\noindent The beamforming vector $\mathbf{v}_t$ is in the null space of $\mathbf{G}^H_{11,f,t-\tau}$, i.e., $\mathbf{G}^H_{11,f,t-\tau} \mathbf{v}_t =0$. Therefore, 

{\small\begin{align}
\nonumber \mathbb{E} \left[ P_t||\mathbf{G}^H_{11,f,t} \mathbf{v}_t||^2 \right]~~~~~~~~~~~~~~~~~~~~~~~~~~~~~~~~~~~~~~~~~~~~~~\\ =
P_t (1-\alpha_f^2) \mathbb{E}\left[||\sum_{\tau' = 0}^{\tau-1} \alpha_f^{\tau - \tau' - 1} \Delta \mathbf{G}^H_{11,f,t-\tau'} \mathbf{v}_t ||^2\right].
\end{align}}

\vspace{-2mm}
\noindent Using the fact that channel evolutions $\Delta \mathbf{G}^H_{11,f,t-\tau'}\sim \mathcal{CN}(0,\mathbf{I})$ are i.i.d., and $\mathbf{v}_t$ is a unit-norm vector, the above equation reduces to
{\small \begin{align}
\nonumber \mathbb{E} \left[ P_t||\mathbf{G}^H_{11,f,t} \mathbf{v}_t||^2 \right] = P_t M_p(1-\alpha_f^2)  \sum_{\tau' = 0}^{\tau-1} \alpha_f^{2(\tau - \tau' - 1)} \\= P_t M_p (1-\alpha_f^{2\tau}),
\end{align}}

\vspace{-5mm}
\noindent where $M_p$ is the rank of $\mathbf{G}_{11,f,t}$.

\section{Expression for $P^{fix}_f = \frac{I^0}{M_p\mathbb{E}_\tau [1-\alpha_f^{2\tau}]}$ }
\label{app:p_n_fix}
Let us define $g(\alpha_f, \mathbf{T}_f) = \mathbb{E}_\tau [1-\alpha_f^{2\tau}]$. The expectation can be written as

{\small \begin{align}
\nonumber g(\alpha_f, \mathbf{T}_f)  =& \sum_{i \in \mathbb{N}} (1 -\alpha_f^{2i }) \times \Pr(\tau=i \vert s_{f,t}=\{1,2\})
\\\nonumber  =& \sum_{i \in \mathbb{N}} (1 -\alpha_f^{2i }) \left[ \frac{\pi_{1,f}}{\pi_{1,f} + \pi_{2,f}} \Pr(\tau=i \vert s_{f,t}=1)\right]
\\&+ \sum_{i \in \mathbb{N}} (1 -\alpha_f^{2i }) \left[ \frac{\pi_{2,f}}{\pi_{1,f} + \pi_{2,f}}  \Pr(\tau=i \vert s_{f,t}=2)\right]
\label{eq:f1}
\end{align}}

\vspace{-2mm}
\noindent where $ \mathbb{N}$ is a set of natural numbers and $\Pr(\tau=i | s_{f,t}=1)$ is the probability that the null space of $i$ slots old when $s_{f,t}=1$. This probability can be written as follows:
{\small \begin{align}
\nonumber \Pr(\tau=i\vert s_{f,t}=1) 
&= \frac{\pi_{2,f}}{\pi_{1,f}} \Pr(s_t=1,\cdots,s_{t-(i-1)}\neq 2 | s_{t-i}=2)
\\&=\frac{\pi_{2,f}}{\pi_{1,f}} \sum_{s\in \{0,1\}} p_{2s}p_{s1\backslash 2}^{(i-1)}.
\label{eq:tau_st_1}
\end{align}}

\vspace{-1mm}
\noindent where $p_{ss' \backslash s''}^{(i)}$ is the probability of PU link going from state $s$ to state $s'$ in $i$ slots without hitting state $s''$. Similarly, 
\begin{align}
\Pr(\tau=i\vert s_{f,t}=2) = \frac{\pi_{1,f}}{\pi_{2,f}}  \sum_{s\in \{0,2\}} p_{1s}p_{s2\backslash 1}^{(i-1)},
\label{eq:tau_st_2}
\end{align}
Substituting (\ref{eq:tau_st_2}) and (\ref{eq:tau_st_1}) in  (\ref{eq:f1}) and then in (\ref{eq:p_n_fix_def}), we get the required expression in (\ref{eq:p_n_fix}).

\section{Proof of theorem \ref{thm:fbdp_v_fbfp}}
\label{app:fbdp_v_fbfp}
Let us consider that SU stays on frequency band $f$ under fixed and dynamic power policies. The expected rates under the two policies can be written as follows:

{\footnotesize
\begin{align}
\nonumber \mathbb{E}[R^{(1)}_{f,t}] =  (1- \pi_{0,f}) \mathbb{E}[R^{(1)}_{f,t} | s_{f,t}=\{1,2\}] + \pi_{0,f} \mathbb{E}[R^{(1)}_{f,t} | s_{f,t}=\{0\}],
\\\nonumber  \mathbb{E}[R^{(2)}_{f,t}] =  (1- \pi_{0,f}) \mathbb{E}[R^{(2)}_{f,t} | s_{f,t}=\{1,2\}] + \pi_{0,f} \mathbb{E}[R^{(2)}_{f,t} | s_{f,t}=\{0\}].
\end{align}}

\noindent Since both SU are on the same band and it transmits same transmit power $P^0$ under the two policies when PU is silent, i.e., $s_{f,t}=0$, we have $\mathbb{E}[R^{(1)}_{f,t} | s_{f,t}=\{0\}]= \mathbb{E}[R^{(2)}_{f,t} | s_{f,t}=\{0\}]$. Therefore, to prove  $\mathbb{E}[R^{(2)}_{f,t}]\geq \mathbb{E}[R^{(1)}_{f,t}]$, it is sufficient to prove that $\mathbb{E}[R^{(2)}_{f,t} \vert s_{f,t}=\{1,2\}]  \geq \mathbb{E}[R^{(1)}_{f,t} \vert s_{f,t}=\{1,2\}] $. Equivalently, we need to prove that

{\small \begin{align}
\nonumber  \mathbb{E}\left[\log_2\left(1 + \frac{I^0}{M_p(1-\alpha_f^{2\tau})} {\Gamma_{f,t}}\right) \vert s_{f,t}=\{1,2\} \right] 
~~~~~~~~~~\\~~\geq \mathbb{E}_{\Gamma}\left[\log_2\left(1 + \frac{I^0}{M_p\mathbb{E}_\tau [1-\alpha_f^{2\tau}]} {\Gamma_{f,t}}\right) \vert s_{f,t}=\{1,2\} \right].
\label{eq:r2_v_r1}
\end{align}}

\noindent The expectation in the LHS can be split in terms of expectation with respect to $\tau$ and $\Gamma$ as:
{\small \begin{align}
\nonumber \mathbb{E}\left[\log_2\left(1 + \frac{I^0}{M_p(1-\alpha_f^{2\tau})} {\Gamma}\right) \vert s_{f,t}=\{1,2\} \right] 
~~~~~~~~~~~~~~~\\ =
 \mathbb{E}_{\Gamma}\left[ \mathbb{E}_{\tau} \left[\log_2\left(1 + \frac{I^0}{M_p(1-\alpha_f^{2\tau})} {\Gamma}\right) \vert s_{f,t}=\{1,2\} \right] \right].
\end{align}}

\noindent For simplicity of notations, we have dropped suffixes from $\Gamma$. In order to prove the inequality in (\ref{eq:r2_v_r1}), it is sufficient to prove that for any given value of $\Gamma_{f,t}$ the following holds:
{\small \begin{align}
\nonumber \mathbb{E}_{\tau} \left[\log_2\left(1 + \frac{I^0}{M_p(1-\alpha_f^{2\tau})} {\Gamma}\right) \vert s_{f,t}=\{1,2\} \right]~~~~~~~\\\geq \log_2\left(1 + \frac{I^0}{M_p\mathbb{E}_\tau [1-\alpha_f^{2\tau}]}{\Gamma}\right).
\label{eq:eqvt}
\end{align}}

\noindent Using Bayes' rule, the term in the LHS can be written as:
{\small
\begin{align}
\nonumber \mathbb{E}_{\tau}\left[\log_2\left(1 + \frac{I^0}{M_p(1-\alpha_f^{2\tau})} {\Gamma}\right) \vert s_{f,t}=\{1,2\} \right] 
~~~~~~~~~~~~~ \\\nonumber = \frac{\pi_{1,f}}{\pi_{1,f} + \pi_{2,f}} \mathbb{E}\left[\log_2\left(1 + \frac{I^0}{M_p(1-\alpha_f^{2\tau})} {\Gamma}\right) \vert s_{f,t}=\{1\} \right]  
\\ + \frac{\pi_{2,f}}{\pi_{1,f} + \pi_{2,f}} \mathbb{E}\left[\log_2\left(1 + \frac{I^0}{M_p(1-\alpha_f^{2\tau})} {\Gamma}\right) \vert s_{f,t}=\{2\} \right].
\label{eq:eqvt2}
\end{align}}

\vspace{-1mm}
Let us define $a = \frac{I^0}{M_p(1-\alpha_f^{2\tau})} | s_{f,t}=1$, i.e., $a$ is a random variable with value $\frac{I^0}{M_p(1-\alpha_f^{2\tau})}$ when $s_{f,t}=1$. Similarly, let 
$b = \frac{I^0}{M_p(1-\alpha_f^{2\tau})} | s_{f,t}=2$. Note that $\Gamma$ follows the same distribution under $s_{f,t}=1$ and $s_{f,t}=2$. Therefore, we can write (\ref{eq:eqvt2}) as follows

{\small \begin{align}
\nonumber \mathbb{E}_{\tau} &\left[\log_2\left(1 + \frac{I^0}{M_p(1-\alpha_f^{2\tau})} {\Gamma}\right) \vert s_{f,t}=\{1,2\} \right] \\&= \frac{\pi_{1,f}}{\pi_{1,f} + \pi_{2,f}} \mathbb{E}_{\tau}\left[\log_2(1+a \Gamma)\right] + \frac{\pi_{2,f}}{\pi_{1,f} + \pi_{2,f}} \mathbb{E}_{\tau} \left[\log_2(1+b \Gamma)\right].
\label{eq:eqvt3}
\end{align}}

\vspace{-1mm}
\noindent Using the property of log function, we have $\frac{\pi_{1,f}}{\pi_{1,f} + \pi_{2,f}} \mathbb{E}_{\tau}\left[\log_2(1+a \Gamma)\right] \geq -\frac{\pi_{1,f}}{\pi_{1,f} + \pi_{2,f}} \log_2 \mathbb{E}_\tau(\frac{1}{1+a\Gamma})$ and $\frac{\pi_{2,f}}{\pi_{1,f} + \pi_{2,f}} \mathbb{E}_{\tau}\left[\log_2(1+b \Gamma)\right] \geq -\frac{\pi_{2,f}}{\pi_{1,f} + \pi_{2,f}} \log_2 \mathbb{E}_\tau(\frac{1}{1+b\Gamma})$.

Now let us define $c = \frac{I^0}{M_p\mathbb{E}_{\tau}[1-\alpha_f^{2\tau}]}$. The term in the RHS of (\ref{eq:eqvt}) is then 
\begin{align}
\log_2\left(1 + \frac{I^0}{M_p\mathbb{E}_\tau [1-\alpha_f^{2\tau}]} {\Gamma}\right) = \log_2(1+ c\Gamma).
\label{eq:eqvt4}
\end{align}
Therefore, using (\ref{eq:eqvt3}), (\ref{eq:eqvt4}), we can say that to prove (\ref{eq:eqvt}) hold for any value of $\Gamma$, it is sufficient to prove the following:
{\small \begin{align}
\nonumber -\frac{\pi_{1,f}}{\pi_{1,f} + \pi_{2,f}} \log_2 \left(\mathbb{E}_\tau \left[\frac{1}{1+a\Gamma}\right] \right) 
~~~~~~~~~~~~~~~~~~~~~~~~~~\\ -\frac{\pi_{2,f}}{\pi_{1,f} + \pi_{2,f}} \log_2 \mathbb{E}_\tau \left( \left[\frac{1}{1+b\Gamma} \right]\right) \geq \log_2 (1 + c\Gamma),
\end{align}}

\noindent or equivalently
{\small \begin{align}
\nonumber  \log_2 \left[ \left(\mathbb{E}_\tau \left[\frac{1}{1+a\Gamma}\right] \right)^{-\pi_{1,f}} \left(\mathbb{E}_\tau \left[\frac{1}{1+b\Gamma}\right] \right)^{-\pi_{2,f}}  \right] \\\geq  \log_2\left[ (1 + c\Gamma)^{{\pi_{1,f} + \pi_{2,f} }}\right].
\end{align}}

\vspace{-1mm} \noindent Since, logarithm is a monotonically increasing function, it is sufficient to prove that
{\small \begin{align}
\nonumber \left(\mathbb{E}_\tau \left[\frac{1}{1+a\Gamma}\right] \right)^{-\pi_{1,f}}  \left(\mathbb{E}_\tau \left[\frac{1}{1+b\Gamma}\right] \right)^{-\pi_{2,f}} ~~~~~~~~~~~~\\\geq   (1 + c\Gamma)^{\pi_{1,f} } (1 + c\Gamma)^{\pi_{2,f} },
\end{align}}

\vspace{-1mm}
\noindent or equivalently
{\small \begin{align}
1 \geq   \left[(1 + c\Gamma)\left(\mathbb{E}\frac{1}{1+a\Gamma}\right)\right]^{\pi_{1,f} }   \left[(1 + c\Gamma)\left(\mathbb{E}\frac{1}{1+b\Gamma}\right)\right] ^{\pi_{2,f} }.
\label{eq:eqvt5}
\end{align}}

\vspace{-1mm}
\noindent Since $0\leq \pi_{1,f}, \pi_{2,f} \leq 1$, the above inequality holds if each term in the square bracket is $\leq 1$. To prove that is the case, we first express the relationship between random variables $a,b$ and $c$ is as $c =  \frac{\pi_{1,f} + \pi_{2,f} }{\mathbb{E}[1/a] + \mathbb{E}[1/b]}$. Therefore, we have
{\begin{align}
(1 + c\Gamma)\mathbb{E}\left[\frac{1}{1+a\Gamma}\right] = \left(1 + \frac{(\pi_{1,f} + \pi_{2,f})\Gamma }{\mathbb{E}[1/a] + \mathbb{E}[1/b]} \right)\mathbb{E}\left[\frac{1}{1+a\Gamma}\right]  
\end{align}}

\vspace{-1mm}
\noindent Note that $a,b,\Gamma \geq 0$ and $0 \leq \pi_{1,f} +\pi_{2,f} \leq 1$. Therefore, to prove that the above term is $\leq 1$, it is sufficient to show that
{\small\begin{align}
\Gamma \mathbb{E} \left[ \frac{1}{1+a\Gamma}\right] \leq \mathbb{E}\left[\frac{1}{a}\right] \left(1- \mathbb{E}\left[\frac{1}{1+a\Gamma}\right] \right)
\end{align}}

\vspace{-1mm}
\noindent Since $\Gamma$ is a constant in the above equation, we can re-arrange the LHS to get the following requirement for the proof:
{\small \begin{align}
\mathbb{E} \left[ \frac{1}{a} \left(\frac{a \Gamma}{1+a\Gamma } \right)\right] \leq  \mathbb{E}\left[\frac{1}{a}\right]\mathbb{E}\left[\frac{a\Gamma}{1+a\Gamma}\right]
\label{eq:eqvt6}
\end{align}}

\vspace{-1mm}
\noindent Note that for any given $\Gamma\geq 0$, the random variables $1/a$ and $a\Gamma /(1+a\Gamma) $ are negatively correlated. Therefore, (\ref{eq:eqvt6}) always holds for any value of $\Gamma$ and (\ref{eq:r2_v_r1}) is always true, which completes the required proof of $\mathbb{E}[R^{(2)}_{f,t}] \geq \mathbb{E}[R^{(1)}_{f,t}]$ (Theorem \ref{thm:fbdp_v_fbfp}).

\bibliographystyle{IEEEtran}
\bibliography{IEEEabrv,final_refs_j3}

\end{document}